\documentclass[journal,10pt,draftclsnofoot,onecolumn]{IEEEtran}
\usepackage{epsfig,psfrag}
\usepackage{color}
\usepackage{amsmath,amssymb}

\newtheorem{theorem}{Theorem}
\newtheorem{lemma}[theorem]{Lemma}
\newtheorem{corollary}[theorem]{Corollary}
\newtheorem{definition}{Definition}
\newenvironment{proof}{\noindent{\bf Proof:}}{
\hspace*{\fill} $\Box$ \vskip \belowdisplayskip}

\newcommand{\beqn}{\begin{equation}}
\newcommand{\eeqn}{\end{equation}}
\newcommand{\ep}{\varepsilon}
\newcommand{\ve}{\varepsilon}
\newcommand{\maxwt}{{\sc Max-Weight} }
\newcommand{\maxwtbt}{{\sc Max-Weight($\beta$)} }
\newcommand{\Ab} {$A(\omega,\epsilon,b)$}
\newcommand{\Ac} {$A(\omega,\epsilon,0)$}

\newcommand{\Al} {Algorithm}
\newcommand{\A} {$A(\omega,\epsilon)$}
\newcommand{\ad} {Adversary}
\newcommand{\om}{\omega}

\newcommand{\red}{\color{black}}
\newcommand{\blue}{\color{black}}
\newcommand{\col}{\color{black}}
\newcommand{\coll}{\color{black}}

\date{}

\begin{document}

\title{Stability of the Max-Weight
Protocol in Adversarial Wireless Networks}

\author{\IEEEauthorblockN{Sungsu Lim,}
\IEEEauthorblockA{KAIST,
ssungssu@kaist.ac.kr}\\
\and
\IEEEauthorblockN{Kyomin Jung,}
\IEEEauthorblockA{KAIST,
kyomin@kaist.edu}\\
\and
\IEEEauthorblockN{Matthew Andrews,}
\IEEEauthorblockA{Bell Labs,
andrews@research.belllabs.com}}

\maketitle

\begin{abstract}
In this paper we consider the \maxwt protocol for routing and
scheduling in wireless networks under an {\em adversarial} model.
This protocol has received a significant amount of attention dating
back to the papers of Tassiulas and Ephremides.  In particular, this
protocol is known to be {\em throughput-optimal} 
whenever the traffic patterns and propagation conditions are governed
by a {\em stationary stochastic process}.

However, the standard proof of throughput optimality (which is based
on the negative drift of a quadratic potential function) does not hold
when the traffic patterns and {\coll the edge capacity changes over time are} governed
by an arbitrary adversarial process. Such an
environment {\coll appears frequently} in many practical wireless scenarios {\coll when} the
assumption that channel conditions are governed by a stationary
stochastic process does not readily apply.

In this paper we prove that even in {\coll the above} adversarial setting, the
\maxwt protocol
{\red keeps the
queues in the network stable (i.e.\ keeps the queue sizes bounded) whenever this is feasible by some routing and scheduling algorithm.}
However, the proof is somewhat
more complex than the negative potential drift argument that applied in
the stationary case.
Our proof holds {\coll for any} arbitrary interference relationships
among edges. We also prove the stability of $\ep$-approximate
\maxwt under the adversarial model.
We conclude the paper with a discussion of queue sizes in the
adversarial model as well as a set of simulation results.
\end{abstract}

\section{Introduction}
We consider the performance of the Max-Weight routing and scheduling
algorithm in adversarial networks. Max-Weight has been one of
the most studied algorithms \cite{AnshelevichKK02, McKeownAW96, MuthukrishnanR98} since it was introduced in the work of
Tassiulas and Ephremides \cite{TassiulasE-92, TassiulasE-93} and Awerbuch and Leighton \cite{AwerbuchT-93, AwerbuchT-94}.
{\blue
The key property
of Max-Weight is that for a fixed set of flows it is {\em throughput optimal} in stochastic
networks with a wide variety of scenarios \cite{TassiulasE-92,Georgiadis06,Shah06}, even
though it may fail to provide maximum stability in a scenario
with flow-level dynamics \cite{Ven09}. That is, for a fixed set of flows the Max-Weight protocol
keeps the queues in the network stable whenever this is
feasible by some routing and scheduling algorithm. Moreover,
we can obtain a bound on the amount of packets in the system
that is polynomial in the network size.}

However, the standard
analyses of the Max-Weight algorithm make critical use of the fact
that the channel conditions and the traffic patterns are governed by
stationary stochastic processes.
The stationary stochastic model deals with the case where traffic patterns do not deviate much from their time-average behavior. On the other hand, we shall consider the worst case traffic scenario modeled by adversarial models. If an adversary chooses traffic patterns and interference conditions, and edge capacities change over time in an arbitrary way, then the question remains as to whether a system running under Max-Weight can be unstable.
It is important to model the worst (adversarial) case because non-evenly distributed traffic patterns are observed over time in many queuing models.
{\blue
A typical adversarial scenario is a military
communication network, in which there could exist adversarial
jammers. Once it is jammed, the victim link will have zero
capacity or very weak capacity. Ensuring stability under the
worst case is crucial in many such systems. The aim of the current paper is to resolve this question.}

Previous work has shed some light on this
issue. In \cite{AndrewsZ-focs02} it was shown that for a single transmitter
sending data over one-hop edges to a set of mobile users, if the set
of non-zero channel rates can approach zero arbitrarily closely, then no
protocol can be stable. However, since this is a fairly unnatural condition, \cite{AndrewsZ-focs02} looked at the more natural setting in which all rate
sets are finite. For this case a stable protocol was given but it was
a somewhat unnatural protocol that relies on a lot of bookkeeping. The
stability of a more natural protocol such as Max-Weight was left
unresolved.

{\coll In some adversarial setting, the stability of
Max-Weight in static networks was proven in \cite{AielloKOR98},
and the stability of Max-Weight was proven in dynamic networks with single-commodity demands \cite{AnshelevichKK02} and multicommodity demands \cite{AndrewsJS-stoc07}.}
However, these proofs only
applied to the case when each edge could be scheduled independently
(in other words, the decision to transmit on an edge has no affect on the
edge rates on other edges),
This is obviously not a suitable
model for wireless transmissions in which edges can clearly affect each
other. As discussed in \cite{AndrewsZ-focs02}, the stability of
Max-Weight was not known in the adversarial setting for the case of
interfering edges, even if we only have one node that transmits.

In this paper we resolve the question of the stability of Max-Weight in
general adversarial networks. We present an adversarial model of
interfering edges and show that the Max-Weight policy always maintains
stability as long as we are strictly within the network stability region,
{\coll
even when the stability region is allowed to change over time.
We consider a very general adversarial model that can be applied to all the possible interference conditions, including k-hop interference \cite{Sharma-06}, independent set constraint \cite{GJSS08, GJSS09}, and node exclusive constraints \cite{AndrewsKRSVW00, AndrewsKRSVW01, BruceS-88, NeelyMR02}.
}
Our proof gives a bound {\coll on the queue size} that is
exponential in the network size. (This is unlike the stochastic case.)
However, we also demonstrate (using an example inspired by
\cite{AndrewsZ-infocom04}) that such exponential queue sizes can
occur.
Although computing the optimal solution of Max-Weight is computationally {\coll NP hard for many scenarios}, in many practical wireless networks $\ep$-approximate
solutions can be computed in polynomial time \cite{Baker94, GJSS08, GJSS09}.
In this paper, we also prove the stability of {\coll any} $\ep$-approximate Max-Weight
under the adversarial model, when $\ep>0$ is small enough. We conclude the paper with a set of
simulation results showing stability of Max-Weight on adversarial setups.

\subsection{Discussion}
We now give a high-level description of the Max-Weight algorithm and
discuss why the standard stochastic analyses are invalid in the
adversarial case. Essentially the protocol operates by maintaining
at each node $v$ a queue of data for each possible
destination $d$. We denote the size of this queue at time $t$ by
$q_{v,d}^t$. For any set of edges in the network, the
{\coll total} {\em weight} on the set at time $t$ is
the sum over all edges {\coll in the set} of the queue differentials multiplied
by the instantaneous edge rates. (A formal definition will be given in
the model section below.)
At all times the \maxwt protocol transmits data on edges so as to
maximize the total weight that it gains. In many situations computing
the exact Max-Weight set of transmissions is a computationally
hard problem. However, in Section 4 we show the stability of an approximate Max-Weight algorithm which can be
implemented efficiently {\coll in many practical setups}.

We say that we are in the stationary stochastic model if there is an
underlying stationary Markov Chain whose state determines the channel
conditions on the edges. We say that we are in the adversarial model
if we do not make such assumptions.
{\coll
In order to make sure that the network is not inherently overloaded the adversarial model assumes that there {\em exists} some way to route and schedule the packets so as to keep the network stable.  However, these routes and schedules are {\em a priori} unknown to the algorithm
}

{\blue
Most previous analyses of Max-Weight have been
performed in a stationary stochastic model} and
they take the following form. Define a quadratic potential
function $P(t)=\sum_{v,d}(q^t_{v,d})^2$ and show, using the assumption that the
traffic arrivals are within the network stability region, that the
potential function always has a negative drift up to an additive
second order term {\coll of $P(t+1)-P(t)=\sum_{v,d}(q^t_{v,d}+\sigma^t_{v,d})^2-\sum_{v,d}(q^t_{v,d})^2$} where $\sigma^t_{v,d}=q^{t+1}_{v,d}-q^t_{v,d}$. Moreover, when the potential function become
sufficiently large, the negative drift in the first order term is
sufficient to overcome the positive second order term. Therefore the
entire potential function has a negative drift. This determines an
upper bound on $P(t)=\sum_{v,d}(q^t_{v,d})^2$ and hence we have an
upper bound on $\sum_{v,d}q^t_{v,d}$.

The reason that this type of analysis does not apply in the
adversarial model is that the channel rates associated with the large
queues in the network may be very small (or even zero).  In this case
we cannot necessarily say that a large queue implies a large drop in the
potential.
Hence for any possible queue configuration there is always
the possibility that the potential function
$P(t)=\sum_{v,d}(q^t_{v,d})^2$ can increase. Hence we need a different
approach to ensure stability.
{\coll
We discuss this in more detail in Section 2.
}

\subsection{Why do adversarial models make sense}
We now briefly discuss why {\coll it is useful} to consider
the adversarial setting which includes the {\em worst} case
scenario;
A model that is governed by a
stationary stochastic process is not general enough to cover many
widely occuring scenarios.
For example, consider a cellular network
in which a car is driving down a road between evenly spaced
basestations.  In this case the channel conditions between the car and
its closest basestation will rise and fall in a periodic
fashion. Moreover, when a car drives into an area of poor coverage
(e.g.\ a tunnel), the channel rate could go to zero. In particular,
this could happen in a haphazard manner that is not modeled by a
stationary stochastic process.

The situation is even more severe in ad-hoc
networks. As nodes move around many of the edges $(i,j)$ will only be active for
a finite amount of time. Hence any stationary stochastic model that gives a
non-zero channel rate to such an edge cannot accurately reflect the
edge rate over a long time period. However, we still wish to ensure
that the queue sizes will not blow up unnecessarily over time and we
believe that an adversarial analysis is one way to address this type
of question.

{\coll In \cite{AndrewsJS-stoc07}, the stability of
Max-Weight in some adversarial model was proven.} However, it was not sufficiently rich to capture
many types of wireless interactions. First of all, in the model of \cite{AndrewsJS-stoc07} all edge rates were either zero or one.
Secondly, when a edge had rate one we could transmit on it regardless of what is happening on the other edges.  However, this model
cannot capture a situation in which edge rates are variable, nor can it capture a scenario with two interfering edges such that
we can transmit on either one in isolation but not both simulataneously.

In this paper we will define a more general adversarial
model in which {\coll any interference conditions are possible and edge rates can vary over time}.
This allows us to capture arbitrary types of wireless interference behavior. In the next section we describe our model in more detail, after which
we present our results.

\subsection{The Model}
We assume a system in which time is divided into {\coll discrete} time slots.
We consider a queueing model for packet transmissions.
Let $D$ be the set of possible destinations.
Each destination in $D$ can be a subset of the set of nodes.
At each time step $t$
a set of feasible edge rate vectors $R(t) \subset \mathbb{R}^k$ is given by the adversary where $k=|E|$, and
$E$ is the set of all directed edges. Suppose that $r(t)\in R(t)$.
It means that if we
write $r(t) = (r_1(t), r_2(t), \ldots, r_k(t)),$ then
it is possible to transmit on edge $e$ at rate $r_e(t)$, {\em for all
edges simultaneously}. In other words we can transmit data of size
$x_1$ on edge $1$, data of size $x_2$ on edge $2$, etc.,\ so long as
$0\le x_e\le r_e(t)$ for all $e$. Note that this means that the rates
satisfy the {\em downward closed} property, i.e.,\ we can always
transmit on an edge at a rate that is less than the rate
{\coll
$r_e(t)$.

{\blue This is a very general setting for the interference model because
it includes all the possible interference constraints,}
including k-hop interference, independent set constraint, and node exclusive constraints.
For example for a dynamic network $G(V,E(t))$, $R(t)=\{(r_e(t))_{e\in E(t)}|r_{e}(t)=0$ or $1$ for all $e\in E(t)$, $r_{e_1}(t)r_{e_2}(t)=0$ if $e_1$ and $e_2$ are incident in $E(t)\}$ represents a set of feasible edge rate vectors of independent set constraints on $E(t)$ that changes over time.

We make the following assumption about the adversary.
(It was shown in \cite{AndrewsZ-focs02} that if we do not have
these conditions then no online protocol can be stable.)
\begin{itemize}
\item All {\coll packet arrival and edge rates} are bounded from above and non-zero rates are bounded away
from zero. In other words, there exist values $R_{\min}>0$ and
$R_{\max}>0$ such that for each $r(t)=(r_1(t),\ldots,r_k(t))\in R(t)$,
$r_e(t)\le R_{\max}$ and if $r_e(t)\neq 0$ then $r_e(t)\ge
R_{\min}$.
\end{itemize}

We now define the $(\omega,\ep)$-adversary.
At each time, it determines the packet arrivals and edge capacities.
Then, the routing and scheduling algorithm decides the packet transfers in the network against the $(\omega,\ep)$-adversary.
{\blue In this manner, our
framework can be understood as a type of sequential game.}

\begin{definition}\label{def:one}
We say that an adversary injecting the packets and controlling the
edges is an $(\omega,\ep)$-adversary, \A, for some $\ep> 0$ and some integer $\omega\ge 1$, called a {\em window} parameter,
if the following holds:
{\coll
The adversary defines the feasible rate vectors and packet arrivals in each time step subject to the constraint that there exists a routing and scheduling algorithm $T$ (possibly involving fractional movement of packets) which keeps the system stable.
Let $t_p$ be the time when a packet $p$ is injected.
Then we can define $\Psi_p=\{(e,t')|t'\in[t_p,t_p+\om-1], \ell(p,e,t')>0,$ where $\ell(p,e,t')$ is a fractional amount of $p$ that is transmitted by $T$ along $e$ at time $t'\}$,
{\blue which corresponds to the movement of
packet $p$ from its source to one of its destinations under the
algorithm $T$.
} For all packet $p$,
$(1-\frac{\ep}{2})$ fraction
\footnote[1]{
In fact, for any $(1-\delta)$ fraction of $p$ with constant $0<\delta<\ep$ all the results in this paper holds.}
of $p$ will arrive to its
destination during the
window $[t_p,t_p+\omega-1]$.}
{\coll For any integer $j$, let
$I^j$ be the set of packets injected during the
window $W_j=[j\om,(j+1)\om-1]$.
Then the adversary assumes that the following holds
$$
\sum_{p\in I^{j}\cup I^{j-1}, (e,t')\in\Psi_p, t'\in W_j} \ell(p,e,t')\le
\sum_{t'\in W_j} (1-\ep)r_e(t'),
$$
where
$r(t')\in R(t')$ are edge rate vectors assigned by $T$.
}
\end{definition}

This is a very general adversarial model because it covers all the possible interference conditions, {\red including k-hop interference, independent set constraint, and node exclusive constraints,} in dynamic networks, and this model includes adversarial models used in \cite{AnshelevichKK02}, \cite{AielloKOR98} and \cite{AndrewsJS-stoc07}.
{\blue We prove the following theorem, which shows that the \maxwt protocol is {\em throughput-optimal} even against the strongest adversary.}

\begin{theorem}\label{cor:one}
The \maxwt protocol is stable under any \A $~$for any $\ep>0$.
\end{theorem}


\subsection{The Protocol}
\label{s:protocol} We now define the \maxwt protocol.
We assume that each node $v$ has $|D|$ queues which correspond to each destination, respectively. Thus, we have $n|D|$ many queues. Let $Q_{v,d}$ be the queue at node $v$ for data
having destination $d$. Let $q^t_{v,d}$ be the total size of data in
queue $Q_{v,d}$ at time $t$.
We define a general routing and scheduling algorithm \maxwtbt
that is parameterized by a parameter $\beta>0$. We use \maxwt to
denote the algorithm with $\beta=1$. In this paper, we will use the term {\em scheduling algorithm} to mean a combined routing and scheduling algorithm.

\vspace{.0in}
\noindent{\textbf \Al ~ \maxwtbt}\ \vspace{.05in} \hrule
\begin{enumerate}
\item Choose $r(t)\in R(t)$ and $d^{(e)}\in D$ for each $e=(v,u)\in E$, such that
$ \sum_{e\in E} s_e(t) \left((q^t_{v,d^{(e)}})^{\beta}
-(q^t_{u,d^{(e)}})^{\beta}\right)$
is maximized (with an arbitrary tiebreaking rule) where
$$s_e(t):=min\left\{r_e(t),\left|\frac{q^t_{v,d^{(e)}}-q^t_{u,d^{(e)}}}{2}\right|\right\}.$$
Send data of size $s_e(t)$
from $Q_{v,d^{(e)}}$ to $Q_{u,d^{(e)}}$ along $e$.

\item For each time $t$, and for each node $v$, accept all packets injected by the \ad $~$to $v$.

\item Remove all packets that arrive at their destination.
\vspace{.05in} \hrule \vspace{.05in}
\end{enumerate}

When $\beta>0$, $(q^t_{v,d^{(e)}})^{\beta}-(q^t_{u,d^{(e)}})^{\beta} \ge 0$ implies $q^t_{v,d^{(e)}}-q^t_{u,d^{(e)}} \ge 0$, so it guarantees all packet movement between
queues occur from a taller queue to a smaller queue.

The algorithm can be understood to be designed so that the following
{\em potential function} decreases as much as possible. (However, as
discussed earlier and unlike in the stochastic case, there is no
simple argument that for sufficiently large queue sizes there always
{\em is} a decrease in potential.)
$$P(t)~\stackrel{\triangle}{=} \sum_{v,d} (q^t_{v,d})^{\beta+1}.$$

\section{Stochastic analysis}
In this section we give more details of the typical stochastic
analysis and explain why this type of analysis does not directly hold
in the adversarial setting. {\blue We say that we are in the
stationary stochastic model if there is an underlying stationary
Markov Chain $\cal M$ with state space $\{m_r\}$ and a function $f(\cdot)$
from $\{m_r\}$ to sets of feasible edge rate vectors $R(t)$ such that
the Markov Chain updates its state at each time step and if it
has state $\{m_r\}$ at time $t$ then $R(t) = f(m_r)$.}

Throughout this section we will focus on the case that $\beta=1$ and
study the potential function $P(t)=\sum_{v,d}(q^t_{v,d})^2$.  Let
$a^t_{v,d}$ (resp.\ $b^t_{v,d}$) be the amount of data arriving into
(resp.\ departing from) $Q_{v,d}$ at time $t$, according to the \maxwt
algorithm.
{\coll
For simplicity we shall also discuss the most basic scenario in which the distribution over feasible service rate vectors is i.i.d.\ at each time step.}
Let ${a'}^t_{v,d}$ and ${b'}^t_{v,d}$ be the corresponding
quantities for the underlying ``optimum'' schedule (that keeps the
system stable by assumption). The expected change in $P(\cdot)$ from time step $t$ to time $t+1$ is given by,
\begin{eqnarray}
&&E[P(t+1) - P(t)]
= E[\sum_{v,d}(q^{t+1}_{v,d})^2 - \sum_{v,d}(q^t_{v,d})^2]
\nonumber \\
&=& E[\sum_{v,d}\left(q^t_{v,d} + a^t_{v,d} - b^t_{v,d})\right)^2 - \sum_{v,d}(q^t_{v,d})^2]
\nonumber\\
&=& E[\sum_{v,d}((q^t_{v,d})^2 + {\coll 2}q^t_{v,d}(a^t_{v,d} - b^t_{v,d})+
(a^t_{v,d} - b^t_{v,d})^2) - \sum_{v,d}(q^t_{v,d})^2]\nonumber\\
&\le & E[\sum_{v,d} ({\coll 2}q^t_{v,d}({a'}^t_{v,d} - {b'}^t_{v,d})+ (a^t_{v,d} - b^t_{v,d})^2)].
\end{eqnarray}
The final inequality is due to the definition of \maxwt since we can
think of \maxwt as always making the decision that minimizes
$\sum_{v,d} q^t_{v,d}({a'}^t_{v,d} - {b'}^t_{v,d})$.
{\coll
By taking into account the i.i.d. nature of the service rate vectors and the fact that the traffic injections can be scheduled by the optimal algorithm,
we have that $E[({a'}^t_{v,d} - {b'}^t_{v,d})]\le -\epsilon$ for all $v,d$.}
since there is an upper bound on the amount of data that can be transfered
between two queues at each time step, $E[(a^t_{v,d} - b^t_{v,d})^2]$ is
bounded by some quantity $C$ that is independent of time.  Hence,
$$
E[P(t+1) - P(t)]
\le  C-{\coll 2}\sum_{v,d} q^t_{v,d}\ve
$$
and thus if there is some $Q_{v,d}$ that satisfies $q^t_{v,d}\ge
C/{\coll 2}\ve$ then the expected drift of $P(t)$ is negative at time $t$. This
in turn implies that $P(t)$ cannot grow indefinitely over time and so
the system is stable.

We can now demonstrate why this type of argument {\em does not} hold
in the adversarial model. In a non-stationary, adversarial environment it is
not necessarily the case that the set $R(t)$ and  packet arrival
rates are independent of the
$q^t_{v,d}$ values. {\coll That is, we cannot assume that a large queue will have good connectivity to the rest of work}, so
there is no analogue of the statement that
$E[({a'}^t_{v,d} - {b'}^t_{v,d})]\le {\coll -}\ve$. In particular, it may be the case
that for all large $q^t_{v,d}$ and for all $r(t)\in R(t)$, the value
of $r_e(t)$ is zero for all edges $e$ that are adjacent to node $v$.
{\coll
Indeed, the fact that we have built up a large queue in one region of the network may be precisely because that region has poor connectivity to other parts of the network.}
Hence we need a different type of argument to show stability in the
adversarial setting and this is the question that we address in this
paper.

\section{Main Results}
At the highest level, our proof proceeds as follows. We first show a result that bears some similarity to the ``negative drift'' result that is used to prove stability in stationary stochastic systems.  In particular in Theorem~\ref{lemma:zero} we show that whenever a packet is injected, we can assign a set of transmissions by the \maxwtbt protocol to the packet such that the resulting decrease in potential almost matches the increase in potential that arises from the packet injection itself. This allows us to bound the increase in potential whenever a packet is injected. (We note as an aside that when there are no packet injections the \maxwtbt protocol ensures that the potential never increases.) Moreover, Theorem~\ref{lemma:zero} also shows that whenever there is an injection to a queue that is sufficiently tall, the assigned transmissions induce a decrease in potential {\em more than} the increase due to the packet injection. Hence for such injections there will always be a decrease in potential.

However, in an adversarial system this type of argument is not sufficient to show stability since it might be the case that most packets are injected into small queues.  We therefore extend the proof of Theorem~\ref{lemma:zero} to a more general result that will ensure stability.  In particular we introduce the notion of a {\em bad injection}.  This is an injection that is extra to the injections that are allowed by our definition of adversary.  This notion is convenient since we will use an inductive proof in which injections to small queues that lead to a big increase in potential are treated as ``extra'' packets by the inductive hypothesis.   In particular, we are able to use an inductive argument to show that the number of bad injections is bounded, and hence we can obtain an upper bound of the potential over all time. This immediately implies the stability of {\scshape Max-Weight($\beta$)}.

We now describe these ideas in a little more detail.
The procedure in our setup is as follows. At each time, an adversary chooses
the packet injections and interference conditions. Then
\maxwtbt determines the (routing and) scheduling of packet transmissions.
To show the stability of {\scshape Max-Weight($\beta$)}, we will define an assignment of each packet with a set of (partial) transmissions in the network, so that any injected packet to a {\em tall} queue will decrease the potential function.

\begin{definition}\label{def:zero}
{\coll We imagine that there are $|D|$ links on each directed edge corresponding to each possible destination respectively.}
Let $L=\{ \ell=(e,d)|e=(v,u)\in E, d\in D\}$ be the set of all links.
Let $p$ be a packet injected at time $t$, and let $W=[t,t+\omega-1]$.
A set of partial transmissions $\Gamma_{p}$ assigned to $p$ is defined as a vector of dimension $\omega|L|$.
Let $s_e(t')$ be the vector chosen by \maxwtbt that maximizes $\sum_{e} s_e(t') \left((q^{t'}_{v,d^{(e)}})^{\beta}-(q^{t'}_{u,d^{(e)}})^{\beta}\right).$
For a given adversary $A(\om,\ep)$, and a scheduling algorithm $Alg$, let $\Gamma_{p}($\A$,Alg)=(s_{p,\ell}(t'))_{\ell\in L, t'\in W}$ be a vector of size $\omega|L|$ that satisfies for each $e\in E$, $d\in D$, $t'\in W$, (i) $s_{p,(e,d)}(t')\ge0,$ and (ii) $\sum_{p,d}s_{p,(e,d)}(t')\le s_e(t').$
We say $\Gamma_{p}($\A$,Alg)$ is a {\em set of (possibly partial)
transmissions assigned with $p$}, for convenience, denote by $\Gamma_{p}$.
\end{definition}

We note that the word {\em partial} is used to reflect the fact that one transmission may correspond to multiple packets $p$ subject to the
condition (ii).
{\red Conceptually, it allows the case that an injected packet can be transmitted to its destination across multiple paths.
Thus, an assignment of partial transmissions $\Gamma_p$ of each packet can represent many general routing patterns.
Moreover, it allows the case when $\Gamma_p$ does not form a set of paths.}
An example of this assignment is shown in Fig 1.
\begin{figure}
\centering
\epsfig{file=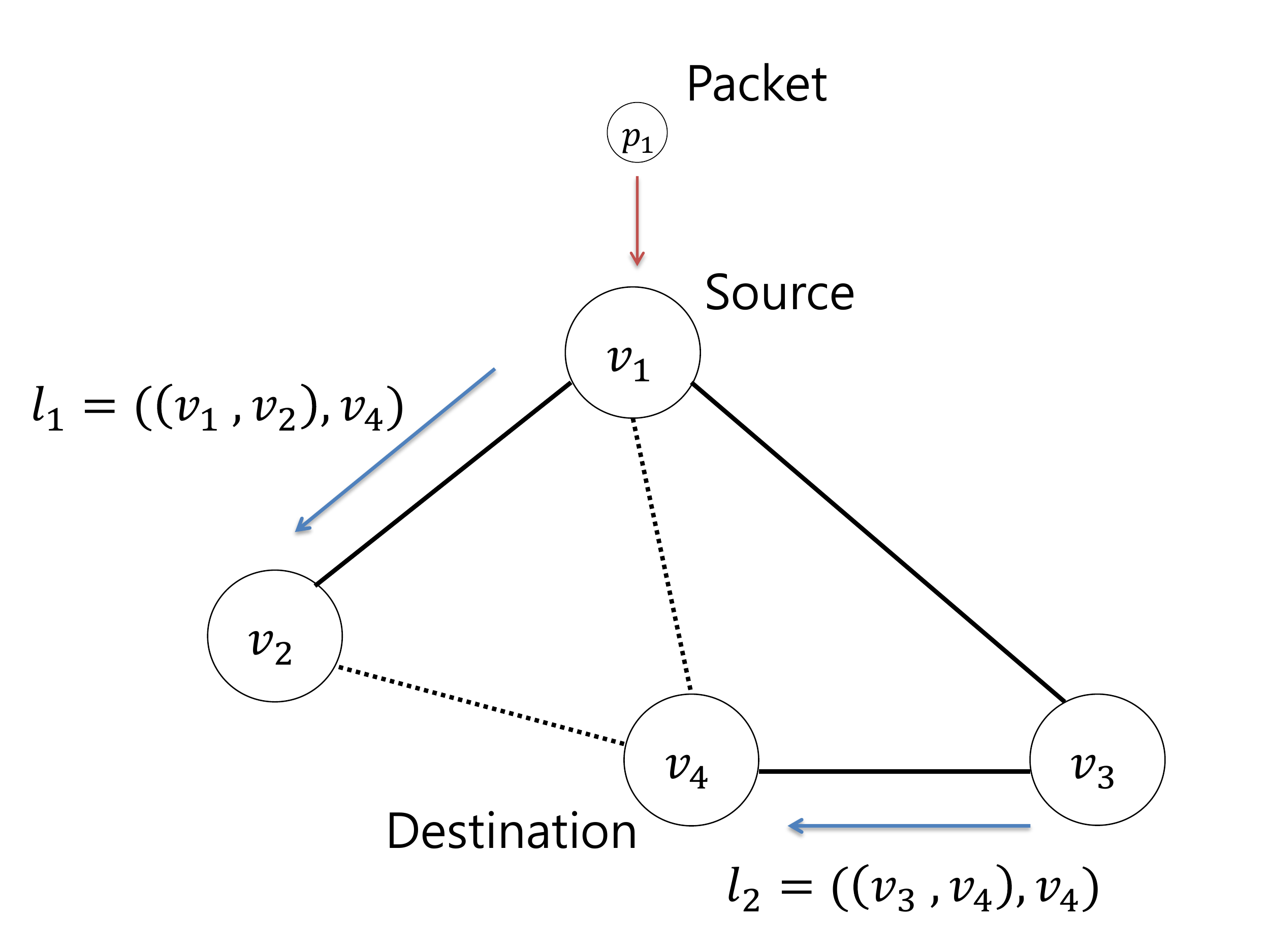, height=1.7in, width=2.2in}
\caption{An example of partial transmission assignment. {\red Suppose that a packet $p_1$ is injected to a node $v_1$ at time $t_{p_1}$, and its destination is $v_4$.
For instance, a set of partial transmissions assigned with $p_1$ that contains
$s_{p_1,\ell_1}(t_{p_1}+2)=0.5$ and $s_{p_1,\ell_2}(t_{p_1}-1)=1$ can be $\Gamma_p$.
Note that the assignments do not need to be at the same time,
and the whole assignments do not need to form a path, or multiple paths.
}
}
\end{figure}
\begin{theorem}\label{lemma:zero}
{\coll
Consider a given adversary \A$~$ for any $\omega\ge 1$ and $\ep>0$,
and the \maxwtbt protocol for some $\beta>0$.
For any injected packet $p$,
we can assign} this packet with
$\Gamma_p=\Gamma_p( A(\omega,\epsilon),~${\scshape Max-Weight($\beta$)}$)$
so that the sum of total potential changes
is less than ${\coll -\frac{\ep}{1-\ep/2}}\ell_p(\beta
+1)q^\beta +\ell_p O(q^{\beta-1}),$ where $q$ is the height of the
queue where the packet is injected. Therefore, there is a constant
$q^*$ depending on $\omega$ and $\ep$, so that if $q\ge q^*$
the sum of potential changes due to the injection is less than $-\frac{\ep}{2}\ell_p q^\beta$.
\end{theorem}

In the next section we will prove the stability of the \maxwt protocol under any $A(\om,\ep)$. The same argument can be applied to prove that the \maxwtbt
protocol with any constant $\beta > 0$ is stable under any \A $~$with $\ep>0$.

Now, we will prove the stability of the \maxwt protocol under any $A(\om,\ep)$. The same argument can be applied to prove that the \maxwtbt
protocol with any constant $\beta > 0$ is stable under any \A $~$with $\ep>0$.

We define a more general adversarial model, which we call a {\em general adversarial queue system} with {\em bad packets}.
%
In this model the number of queues can be any finite number, not only of the form $n|D|$.
%
An adversary allowing $b$ many bad packets is defined as follows.

{\col
By Theorem~\ref{lemma:zero}, for any \A $~$under {\scshape Max-Weight}, for each injection
of packet $p$, we can assign this packet with a set of partial
transmissions $\Gamma_p$ so that the sum of potential changes due
to these movements are at most ${\coll -\frac{\ep}{1-\ep/2}}\ell_p q+C,$
where $q$ is the height of the queue where the packet is injected
and $C$ is a constant depending on $\om$ and $\ep$
(but not on $n$ and $t$).
In a general adversarial queueing system, we also consider the same assignment $\Gamma_p$.
If the sum of potential changes due to $\Gamma_p$
is at least ${\coll -\frac{\ep}{1-\ep/2}}\ell_p q+C+1$, we now say this a
{\em bad packet}.
We say all the other injected packets are {\em good packets}.

\begin{definition}\label{def:ab}
We say that an adversary injecting the packets and controlling edge capacities in a general adversarial queue system is an \Ab $~$for some
$\ep> 0$ and some integers $\omega\ge 1$ and $b\ge 0$, if the
following holds: there exists a scheduling algorithm $Alg$ and an assignment of
partial transmissions for each injected packet $p$
(for example, the collection of $\Gamma_p$ for \maxwt protocol),
such that among all the packets injected over all time, there are at most $b$ bad packets.
\end{definition}

In the proof of \maxwt stability, we will use an induction on the number of queues.
For a given subset of queues, we can imagine a smaller (sub-)system of those queues.
For an injected packet $p$, if {\em too much} of the assigned partial transmissions do not occur between the queues of the sub-system, we will consider $p$ as a bad packet.
In the analysis, we will use the following property of good packets. 

\begin{lemma}\label{lemma:zero2}
Consider a general adversarial queue system 
\Ab$~$ with a corresponding scheduling algorithm $Alg$ and a
corresponding set of partial transmissions of packets $\Gamma$.
Then there is a constant $q^*$ depending on $\omega$ and $\ep$,
so that for any good packet $p$ injected to a queue of height $q$, if $q\ge q^*$ the sum of the {\em decrease of potential}
due $p$ is more than $\frac{\ep}{2} \ell_p q$.
\end{lemma}
}

\begin{proof}
From the definition, the sum of potential changes due to the injection
of any good packet $p$ is at most $-\frac{\ep}{1-\ep/2}\ell_p q+C+1$.
Let $q^*=\frac{4-2\ep}{(2\ep+{\ep}^2)(\ell_p)}(C+1)$,
then for any $q\ge q^*$,
$\frac{\ep}{1-\ep/2}\ell_p q -\frac{\ep}{2}\ell_p q$
$= \frac{2\ep+{\ep}^2}{4-2\ep}\ell_p q$
$\ge \frac{2\ep+{\ep}^2}{4-2\ep}\ell_p q^*=C+1$.
Thus, the decrease of potential is more than $\frac{\ep}{2}\ell_p q$ for $q\ge q^*$.
\end{proof}

{\col
The crux of our analysis will involve proving the following results (in section 4.2).

\begin{theorem}\label{thm:two}
Consider any general adversarial queue system \Ab $~$for any constant $\ep>0$ with corresponding scheduling algorithm $Alg$.
If $Alg$ {\col guarantees all packet movement between
queues occurs from a taller queue to a smaller queue},
then $Alg$ is stable.
\end{theorem}

Hence from Theorem~\ref{thm:two} we obtain Theorem~\ref{cor:one} directly.
}

\section{Proofs of Theorems}
\subsection{Proof of Theorem~\ref{lemma:zero}}

\begin{proof}
We divide time into windows of $\omega$ time steps, $[0,\omega-1],
[\omega,2\omega-1], [2\omega,3\omega-1], \ldots$.
Since $W_j=[j\omega,(j+1)\omega-1]$
for all integer $j\ge 0$, the collection of $W_j$ for $j\ge 0$
is non-overlapped and the union of this collection covers
all time slots $t$.
From now on let $W=W_j$ for some integer $j\ge 0$.

For each time $t'\in W$, and for each node $v$,
we accept all packets injected by the adversary.
For each packet $p\in {\coll I^j \cup I^{j-1}}$ we will associate some fraction $r_p=\{d_{p,e}(t')|(e,t')\in \Psi_p\}$ of
rates of directed edges used in $\Psi_p$ as follows.
Let $p_1,\ldots p_m$ be all the packets injected in ${\coll I^j \cup I^{j-1}}$. {\coll The order of $p_i$'s can be any possible ordering. Then from the definition 1,}
\begin{equation}\label{eq:b}
\sum_{\begin{subarray}{c} i=1,\ldots,m,
t'\in\{t'\in W|(e,t')\in \Psi_{p_i}\}
\end{subarray}} {\coll \ell(p_i,e,t')}\le \left(1-\ep\right)\sum_{t'\in
W}r_e^{(0)}(t').
\end{equation}
{\coll where $r^{(0)}(t)\in R(t)$ are edge rate vectors assigned by $T$.}

First, for $p_1$ and for each {\col directed edge} $e$ used in $\Psi_{p_1}$, we can define
$d_{p_1,e}(t')$ for each $t'\in W$ so that
\begin{equation}\label{eq:c}
0\le d_{p_1,e}(t')\le
r_e(t')~ ~ \mbox{and} ~ ~ (1-\ep)\sum_{t'\in W} d_{p_1,
e}(t')={\coll \sum_{t'\in W} \ell(p_1,e,t')}.
\end{equation}

{\col We define $d_{p_1,e}(t')=0$ for all directed edge $e$ which is not used in $\Psi_{p_1}$.}
Then, from (\ref{eq:b}) and (\ref{eq:c}), for each $e\in E$ and $t'\in W$, let
$r_e^{(1)}(t')=r_e(t')-d_{p_1,e}(t')$. Then we have,
$$
\sum_{i=2, (e,t')\in \Psi_{p_i}, t'\in W}^m {\coll \ell(p_i,e,t')}\le (1-\ep)\sum_{t'\in
W}r_e^{(1)}(t').
$$
Similarly for $p_2$ and for each $e$ used in $\Psi_{p_2}$ we can define
$d_{p_2,e}(t')$ for each $t'\in W$ so that $$0\le d_{p_2,e}(t')\le
r_e^{(1)}(t')~ ~ \mbox{and} ~ ~ (1-\ep)\sum_{t'\in W} d_{p_2,
e}(t')={\coll \sum_{t'\in W} \ell(p_2,e,t')}.$$

By continuing this process, we can define $d_{p_i,e}(t')$
{\col inductively for all $i\ge2$}, for each $e$ used in $\Psi_{p_i}$ and $t'\in W$ so that
\begin{equation}\label{eq:cc}
0\le d_{p_i,e}(t')\le
r^{(i-1)}_e(t')~ ~ \mbox{and} ~ ~ (1-\ep)\sum_{t'\in W} d_{p_i,
e}(t')={\coll \sum_{t' \in W} \ell(p_i,e,t')}.
\end{equation}

At time $t'$, think of a {\col directed edge} $e=(v,u)\in E$, a {\col link} $(e,d)$, and suppose that $e$ has rate
$r_e(t')$ at time $t'$, and $q^{t'}_{v,d}\ge q^{t'}_{u,d}+r_e(t')$. Then the
potential change $C_e(t')$ due to transmission via {\col a link} $(e,d)$ at time $t'$ is
\begin{eqnarray}\label{eq:d}
C_e(t')=&(q^{t'}_{u,d}+{\coll r_e(t')})^{\beta+1}-(q^{t'}_{u,d})^{\beta+1}
+(q^{t'}_{v,d}-{\coll r_e(t')})^{\beta+1}-(q^{t'}_{v,d})^{\beta+1}
\nonumber \\
=&r_e(t')(\beta+1)\left((q^{t'}_{u,d})^\beta-(q^{t'}_{v,d})^\beta\right)
+r_e(t') O\left((q^{t'}_{u,d})^{\beta-1}+(q^{t'}_{v,d})^{\beta-1}\right).
\end{eqnarray}

Note that this is also true when $|q^{t'}_{u,d}-q^{t'}_{v,d}|<r_e(t')$.
Hence, when $d_{p,e}(t')$ amount of edge rate of $e$ at time $t'$
is assigned to an injected packet $p$, we can consider
$d_{p,e}(t')(\beta+1)\left((q^{t'}_{u,d})^\beta-(q^{t'}_{v,d})^\beta\right)
+d_{p,e}(t')
O\left((q^{t'}_{u,d})^{\beta-1}+(q^{t'}_{v,d})^{\beta-1}\right)$
amount of potential change is induced by a packet $p$.

{\col
We consider the sum of potential changes at each time $t'$ by {\scshape Max-Weight($\beta$)}. Let $s_e(t')$ be a vector chosen by {\scshape Max-Weight($\beta$)}.
From (\ref{eq:d}),
\begin{equation}\label{eq:e}
C_e(t') \ge s_e(t')(\beta+1)\left((q^{t'}_{u,d^{(e)}})^\beta-(q^{t'}_{v,d^{(e)}})^\beta\right)
-R_{\max} O\left((q^{t'}_{u,d^{(e)}})^{\beta-1}+(q^{t'}_{v,d^{(e)}})^{\beta-1}\right).
\end{equation}

From (\ref{eq:e}), we obtain that
\begin{eqnarray}\label{eq:f}
\sum_{e\in E}C_e(t') \ge \sum_{e\in E} s_e(t')(\beta+1)\left((q^{t'}_{u,d^{(e)}})^\beta-(q^{t'}_{v,d^{(e)}})^\beta\right)
-R_{\max} O\left((q^{t'}_{u,d^{(e)}})^{\beta-1}+(q^{t'}_{v,d^{(e)}})^{\beta-1}\right).
\end{eqnarray}

Thus, if we fix the time $t'$, then the sum of potential changes at $t'$ by \maxwtbt is less than or equal to the sum of potential changes at $t'$ by $d_{p,e}(t')$.
We want to define $\Gamma_p$ so that the sum of potential changes by
$s_{p,(e,d)}(t')$ is equal to the sum of potential changes by $d_{p,e}(t')$.

Firstly, we fix $t'\in W$.
Let $p_1,\ldots,p_m$ be the packets injected in $I^W$.
Let $E=\{e_1,\ldots,e_k\}.$ The order of $e_i$'s can be any possible
ordering. For each $e_j\in E$, let
\begin{equation}\label{eq:g}
K_{e_j}(t')=\sum_{i=1}^{m} d_{p_i,(v_j,u_j)}(t')\left((q^{t'}_{v_j,d_i})^\beta-(q^{t'}_{u_j,d_i})^\beta\right)
\end{equation}
where $d_i$ is the destination of $p_i$.
Let
\begin{equation}\label{eq:g1}
J(t')=\sum_{j=1}^{k} s_{e_j}(t')\left((q^{t'}_{v_j,d^{(e_j)}})^{\beta}-(q^{t'}_{u_j,d^{(e_j)}})^{\beta}\right)
\end{equation}
where $e_j=(v_j,u_j)$.
At first, we define
\begin{equation}\label{eq:g2}
s_{p_1,(e_1,d_1)}(t')=min\left\{\frac{J(t')}{(q^{t'}_{v_1,d_1})^{\beta}-(q^{t'}_{u_1,d_1})^{\beta}},
s_{e_1}(t'),
\frac{K_{e_1}(t')}{(q^{t'}_{v_1,d_1})^\beta-(q^{t'}_{u_1,d_1})^\beta}\right\},
\end{equation}
if $s_{e_1}(t')> 0$, and $s_{p_1,(e_1,d_1)}(t')=0$ otherwise. Since $s_{e_1}(t')$ is chosen by {\scshape Max-Weight($\beta$)}, $s_{p_1,(e_1,d_1)}(t')\ge0$.
Next, let $s_{e_1}^{(1)}(t')=s_{e_1}(t')-s_{p_1,(e_1,d_1)}(t'),$
$J^{(1)}(t')=J(t')-s_{p_1,(e_1,d_1)}(t')\left((q^{t'}_{v_1,d_1})^{\beta}-(q^{t'}_{u_1,d_1})^{\beta}\right),$
and $K_{e_1}^{(1)}(t')=K_{e_1}(t')-s_{p_1,(e_1,d_1)}(t')\left((q^{t'}_{v_1,d_1})^{\beta}-(q^{t'}_{u_1,d_1})^{\beta}\right).$

Similary, for all $2\le i \le m$, we can define
\begin{equation}\label{eq:g3}
s_{p_i,(e_1,d_i)}(t')=min\left\{\frac{J^{(i-1)}(t')}{(q^{t'}_{v_1,d_i})^{\beta}-(q^{t'}_{u_1,d_i})^{\beta}},
s_{e_1}^{(i-1)}(t'), \frac{K_{e_1}^{(i-1)}(t')}{(q^{t'}_{v_1,d_i})^\beta-(q^{t'}_{u_1,d_i})^\beta}\right\},
\end{equation}
if $s_{e_1}^{(i-1)}(t')> 0$, and $s_{p_i,(e_1,d_i)}(t')=0$ otherwise.
Let $s_{e_1}^{(i)}(t')=s_{e_1}^{(i-1)}(t')-s_{p_i,(e_1,d_i)}(t'),$
$J^{(i)}(t')=J{\coll ^{(i-1)}}(t')-s_{p_i,(e_1,d_i)}(t')\left((q^{t'}_{v_1,d_i})^{\beta}-(q^{t'}_{u_1,d_i})^{\beta}\right),$
and $K_{e_1}^{(i)}(t')=K_{e_1}^{(i{\coll -1})}(t')-s_{p_i,(e_1,d_i)}(t')\left((q^{t'}_{v_1,d_i})^{\beta}-(q^{t'}_{u_1,d_i})^{\beta}\right).$

Now, from (\ref{eq:f}),
we can define inductively $s_{p_i,(e_j,d_i)}$ for all $j=2,\ldots,k,$
and $i=2,\ldots,m,$ so that
\begin{equation}\label{eq:g4}
s_{p_i,(e_j,d_i)}(t')=min\left\{\frac{J^{((j-1)m+(i-1))}(t')}{(q^{t'}_{v_j,d_i})^{\beta}-(q^{t'}_{u_j,d_i})^{\beta}},
s_{e_j}^{(i-1)}(t'), \frac{K_{e_j}^{(i-1)}(t')}{(q^{t'}_{v_j,d_i})^\beta-(q^{t'}_{u_j,d_i})^\beta}\right\},
\end{equation}
if $s_{e_j}^{(i-1)}(t')> 0$, and $s_{p_i,(e_j,d_i)}(t')=0$ otherwise,
where $e_j=(v_j,u_j),$ and $d_i$ is the destination of $p_i$.
}

{\col
Let $\Gamma_{p}=(s_{p_i,(e_j,d_i)}(t'))_{e_j\in E, t'\in W}$, where $d_i$ is the destination of $p_i$. for each $p_i \in I^W$.
We obtain that
$s_{p,(e,d)}(t')\ge0,$ $\sum_{p,d}s_{p,(e,d)}(t')\le s_e(t'),$
\begin{eqnarray}\label{eq:1111}
\sum_{j=1}^{k}\sum_{i=1}^{m} s_{p_i,(e_j,d_i)}(t')\left((q^{t'}_{v_j,d_i})^\beta-(q^{t'}_{u_j,d_i})^\beta\right)
&{\coll \le}&J(t')
\nonumber\\
&{\coll =}& \sum_{j=1}^{k}s_{e_j}(t')\left((q^{t'}_{v_j,d^{(e_j)}})^\beta-(q^{t'}_{u_j,d^{(e_j)}})^\beta\right),
\end{eqnarray}
and also the followings holds.
\begin{eqnarray}\label{eq:222}
\sum_{j=1}^{k} \sum_{i=1}^{m} s_{p_i,(e_j,d_i)}(t') \left((q^{t'}_{v_j,d_{\coll i}})^\beta-(q^{t'}_{u_j,d_{\coll i}})^\beta\right)
&{\coll =}& \sum_{j=1}^{k} K_{e_j}(t')
\nonumber\\
&=& \sum_{j=1}^{k} \sum_{i=1}^{m} d_{p_i,e_j}(t')\left((q^{t'}_{v_j,d_i})^\beta-(q^{t'}_{u_j,d_i})^\beta\right).
\end{eqnarray}

In the previous assignment, we first defined $s_{p_1,(e_j,d_1)}$ for $j=1,\ldots,k$.
From (\ref{eq:1111}) and (\ref{eq:222}), \maxwt algorithm guarantees that the following inequalities hold.
\begin{eqnarray}\label{eq:33}
\sum_{j=1}^{k} s_{p_i,(e_j,d_i)}(t') \left((q^{t'}_{v_j,d_{\coll i}})^\beta-(q^{t'}_{u_j,d_{\coll i}})^\beta\right)
&\le&\sum_{j=1}^{k} \sum_{i=1}^{m} s_{p_i,(e_j,d_i)}(t') \left((q^{t'}_{v_j,d_i})^\beta-(q^{t'}_{u_j,d_i})^\beta\right)
\nonumber\\
&=& \sum_{j=1}^{k} \sum_{i=1}^{m} d_{p_i,e_j}(t')\left((q^{t'}_{v_j,d_i})^\beta-(q^{t'}_{u_j,d_i})^\beta\right)
\nonumber\\
&\le&\sum_{j=1}^{k}s_{e_j}(t')\left((q^{t'}_{v_j,d^{(e_j)}})^\beta-(q^{t'}_{u_j,d^{(e_j)}})^\beta\right)
\end{eqnarray}

Thus, we can assign $s_{p_i,(e_1,d_i)(t')}$ for $i=1,\ldots,m,$ so that
$\sum_{i=1}^{m}s_{p_i,(e_1,d_i)(t')}\big((q^{t'}_{v_1,d_i})^\beta-(q^{t'}_{u_1,d_i})^\beta\big)=K_{e_1}(t').$
Similary, for all $j\ge2,$ we can assign $s_{p_i,(e_j,d_i)(t')}$ for $i=1,\ldots,m,$ so that
$\sum_{i=1}^{m}s_{p_i,(e_j,d_i)(t')}\big((q^{t'}_{v_j,d_i})^\beta $ $ -(q^{t'}_{u_j,d_i})^\beta\big)=K_{e_j}(t').$
{\col Then by taking the sum of the above inequallities we derive that}
strictly equallity holds in (\ref{eq:222}).
{\col So $\Gamma_p$ is well-defined by the $s_{p,(e,d)}$ values.}
Thus we assigned all packet $p \in I^W$ with $\Gamma_p$ so that
the assigned amount of partial packet transmissions in each link
at time $t'$ is less than or equal to the amount of packet
transmissions of \maxwtbt in each link at time $t'$.
}

The proof that for any \A$~$under \maxwtbt, for each injection
of a packet $p$,
the sum of potential changes due to the injection of $p$ and $\Gamma_p$ is at most ${\coll -\frac{\ep}{1-\ep/2}}\ell_p(\beta
+1)q^\beta +\ell_p O(q^{\beta-1}),$
where $q$ is the height of the queue where the packet is injected,
is in the Appendix A.
\end{proof}

\subsection{Proof of Theorem~\ref{thm:two}}
\begin{proof}
Let $\ep>0$ and let $\omega \ge 1$ be {\col some integer.
Consider a general adversarial queue system \Ab$~$ with scheduling algorithm $Alg$.
Let $n$ be the number of queues in this system.}
We will show that there is a constant
$U(n,q_0,b)$ such that for {\col \Ab}, when the size of the tallest queue
at time $t=0$ is at most $q_0$, the sizes of all queues over all $t\ge
0$ is bounded above by $U(n,q_0,b)$.

We {\col induct} on $n$ to show that for any $q_0\ge 0$ and
$b\ge 0$, there exists $U(n,q_0,b)$. {\col For the basic step,} when $n=1$, {\col there is only one queue in the system, and thus it should be a destination queue. Hence, $U(n,q_0,b)$ exists.}

{\col For the inductive step, we assume} that there is $U(m,q_0,b)$ for all $1\le m\le n-1$, and
for all $q_0\ge 0$ and $b\ge 0$. Using this induction hypothesis, we
will show that for any $q_0$, $U(n,q_0,0)$ exists.
We can set {\col
$U(n,q_0,1)=U(n,U(n,q_0,0)+R_{\max},0),$
because} at each time when the bad packet arrives,
the size of the tallest queue is at
most $U(n,q_0,0)$ {\col and we can transmit data of size at most
$R_{\max}$ on each link.} Similarly {\col for any $i\ge 1$,} we can set
{\col
\begin{equation} \label{eq:one}
U(n,q_0,i)=U(n,U(n,q_0,i-1)+R_{\max},0)\end{equation}
}
by considering the time when the
$i$th bad packet arrives.
Now we only need to prove that $U(n,q_0,0)$ exists.
Let $P(t)$ be the potential of the queues at time $t$.
{\col Note that each injection to a queue of size at most $q^*$
makes the potential increase by at most
$(2R_{\max}q^*+R_{\max}^2)$.
By Theorem~\ref{lemma:zero},} the maximum possible increase of potential
induced by all injections during any time window of size $\om$ is
bounded by some constant $P_0$. Now, for fixed $n$, we define the following; Let $M_{n}=0$. Given $M_{k+1}$, for $j=1,2,\ldots, k$,
define{\col
$$S_{j}~ \stackrel{\triangle}{=}U\left(n-k,M_{k+1},\frac{(j-1)}{2}(L_{1}+L_{2}\ldots+L_{j-1})^2\right),$$
$$L_{j}~ \stackrel{\triangle}{=}\frac{{\coll 2}(n-k) S_{j}^2}{\ep},$$
$$M_{k}~ \stackrel{\triangle}{=}\frac{L_{k}}{R_{\min}}+S_{k}+\frac{{\coll 2}P_0}{\ep R_{\min}}.$$}
Then $M_{k}$, $k=1,2,\ldots, n$, are decreasing over
$k$ {\col($M_1\gg M_2\gg \ldots \gg M_n=0$). We will} show that
for any \Ac $~$ for a general adversarial queue system with $n$
queues, for all time $t\ge 0$, $P(t)$ is bounded by some value that is
independent of $t$. More precisely we will show that{\col
\begin{equation}\label{eq:ineq}
P(t)\le(n-1)M_{1}^2+\max\{n q_0^2,
nM_{1}^2+2\sqrt{n}R_{\max}
M_{1}+R_{\max}^2\}.\end{equation}
Note that the right-hand side of (\ref{eq:ineq}) is independent of $t$, so we can conclude that $U(n,q_0,0)$ exists.}

Now suppose that we are given a general adversarial queue system
with $n$ queues controlled by an \Ac $~$ and {\col some given
scheduling algorithm $Alg$ and a corresponding set $\Gamma$ of partial transmissions assigned with packets }, such that all
the initial queue sizes are at most $q_0$.

Suppose that for all time $t$, $P(t)< nM_{1}^2$ {\col holds}.
Then it {\col implies} that the given scheduling algorithm {\col $Alg$ is stable and (\ref{eq:ineq}) is satisfied.}
Now suppose that there is $t_0$ such that $P(t_0)\ge nM_{1}^2$.
By choosing the smallest such $t_0$, we may assume that {\col
$P(t_0)\le \max\{n q_0^2, nM^2_{1}+2\sqrt{n}R_{\max}
M_{1}+R_{\max}^2\}$} since if
$P(t_0-1)<nM_{1}^2$, the change of potential between time $t_0-1$ and
$t_0$ is at most $2\sqrt{n}R_{\max}
M_{1}+R_{\max}^2$. Note that {\col if $P(t_0)\ge
nM_{1}^2$, then there is a queue of size at least $M_1$, so}
the size of tallest queue at that time is at least $M_1$.

Let $q_1\ge q_2 \ge\ldots \ge q_n=0$ be the ordered sizes of the queues
at time $t_0$. For $1\le j\le n$, let $Q_j$ be the corresponding
$j$th tallest queue at time $t_0$. Then since $q_1\ge M_{1}$ and
$q_n=M_n=0$, there exists some $1\le
k\le (n-1)$ such that $q_k\ge M_{k}$ and $q_{k+1}\le M_{k+1}$.
{\col Hence, $q_k\gg q_{k+1}$ and the sizes of the small queues
stay much smaller than $q_k$, and so the sizes of the tall queues are much
bigger than those of the small queues. We will show that} for all
the time afterward the size of the $(k+1)$th tallest queue stays much
smaller than $M_{k}$. A precise description will appear later.

Now fix one such $k$. We will {\col say} all the queues having size at
least $M_{k}$ at time $t_0$ are {\em``tall queues''}, and all the other
queues {\em ``small queues''}. {\col Recall that by our assumption on the \maxwt protocol, data from a small queue will never move to a tall queue.} Hence we can
consider the set of all the small queues as a separate general
adversarial queue system. We will call this queue system a {\em
system of small queues}. {\col Afterward,} we will use an inductive
argument on this system of small queues to guarantee that their sizes
are bounded by {\col a constant $S_j$ for some $1\le j\le k$ during
some period of time.}

\begin{figure}
\centering
\epsfig{file=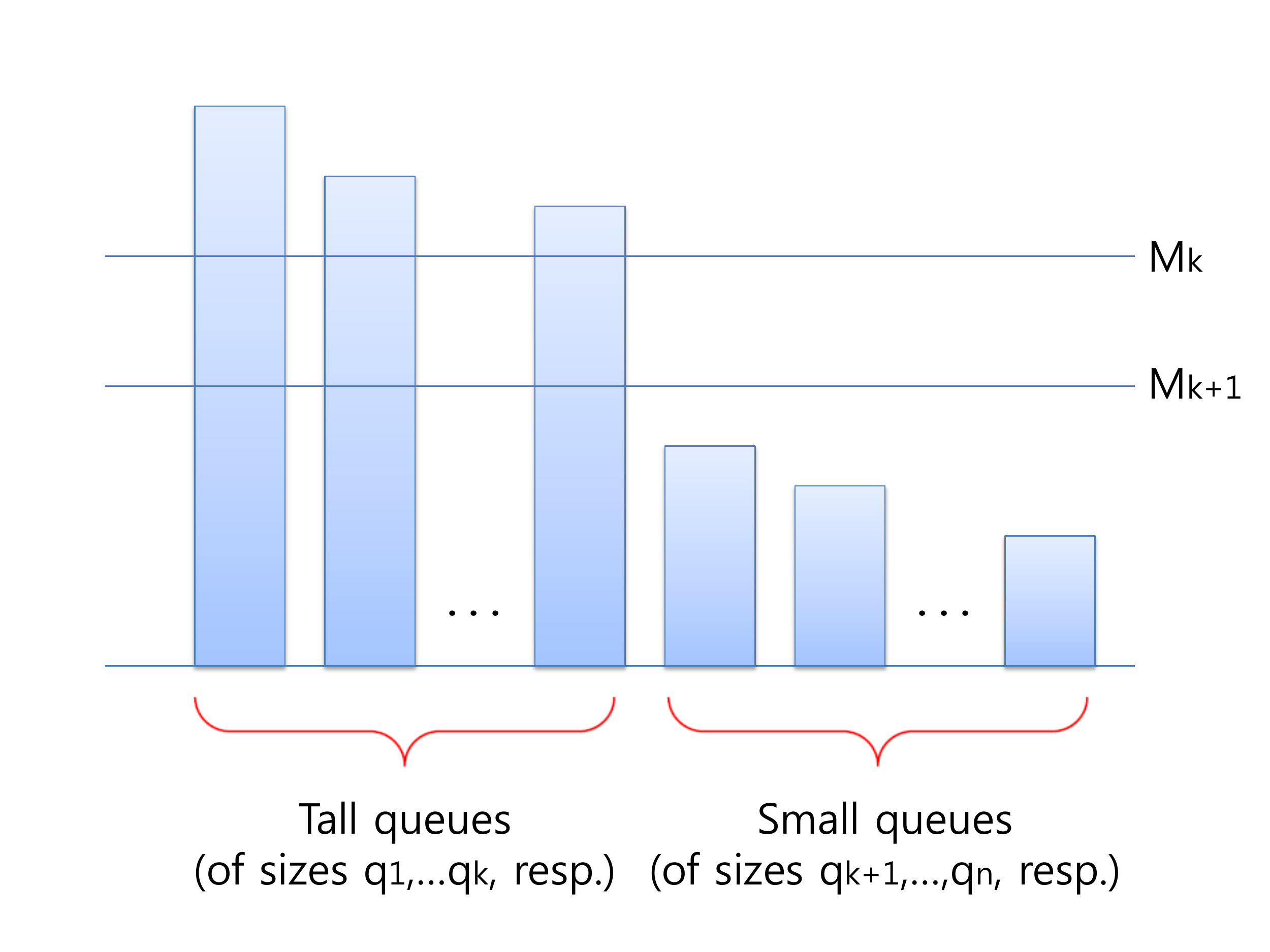, height=2.4in, width=3in}
\caption{All the queues having size at
least $M_{k}$ at time $t_0$ are called {\em``tall queues''}
and all the other queues are called {\em ``small queues''}.
Then, tall queues are much higher than small queues.}
\end{figure}

Let $t_1$ be the first
time after $t_0$ such that there is an injection of a packet to a
tall queue or a transmission of a packet from a tall queue to a
small queue.
Here we note that in the case when there is no such
$t_1>t_0$, then for this \Ac, the argument that will be presented in
the proof of Lemma \ref{lemma:two2} shows that the sizes of all the
small queues cannot be bigger than $M_{k}$ for any time $t\ge t_0$.
{\col Since a packet in a small queue will never move to a tall queue}
, the potential of tall queues are
non-increasing over all time. Hence we obtain that $P(t)$ is bounded
by {\col $(n-1)M^2_{k}+P(t_0)\le (n-1)M^2_{1}+\max\{n q_0^2,
nM^2_{1}+2\sqrt{n}R_{\max}M_{1}+R_{\max}^2\}$} for all $t\ge t_0$ as required in
(\ref{eq:ineq}).

When there is such a $t_1$, our main argument is that during time
$t_0\le t \le t_1$, {\col the system of small queues is maintained.}
{\col
By Lemma \ref{lemma:zero2}, we are able to show a net decrease in
the potential in the system, as long as there are ``sufficient''
injection into queues that are large enough.} Hence, one injection
to a tall queue or one transmission of a packet from a tall queue to
a small queue creates a sufficient decrease in potential.
{\col
We can therefore show that the potential remains bounded as long as
the increase in potential between times $t_0$ and $t_1$ is less than
the decrease in potential due to the injection or transmission at
time $t_1$.} We will prove the following Lemma.

\begin{lemma}\label{lemma:two2} There is $t^*$, satisfying $t_0<t^*\le t_1+\omega-1$, such that $P(t^*)\le P(t_0)$, and during $t_0\le t\le t^*$ the sizes of small queues are
bounded by $M_{k}$.
\end{lemma}

{\col The proof of Lemma \ref{lemma:two2} will appear later after we conclude the proof of Theorem~\ref{thm:two}.
By applying this,} the potential of all the small queues,
$P_{S}(t)~ \stackrel{\triangle}{=} \sum_{i=k+1}^n (q^{t}_{i})^2,$ is
bounded above by $(n-1)M_{k}^2 \le (n-1)M_{1}^2$ since the sizes of
all the small queues cannot be bigger than $M_{k}$ for any time
$t_0\le t\le t^*-1$. Note also that until the
time $t^*-1$, the potential of all the tall queues, $P_{T}(t)~
\stackrel{\triangle}{=} \sum_{i=1}^k (q^{t}_{i})^2,$ is
non-increasing over time.
Since at time $t^*$ we know that the total potential
$P(t^*)\le P(t_0)$, for $t_0\le t\le t^*$,
the potential $P(t)$  is bounded by {\col
$(n-1)M_{1}^2 + \max\{nq_0^2,nM^2_{1}+2\sqrt{n}R_{\max}M_{1}+R_{\max}^2\}$ }
We now choose the
first time $t\ge t^*$, if there exists such $t$, so that
$P(t)\ge nM_{1}^2$, and set this time as a new $t_0$.
Then by applying the same argument, we obtain that for all time
$t\ge 0$, (\ref{eq:ineq}) holds. Hence $U(n,q_0,0)$ exists.
{\col
It implies (\ref{eq:one}) which in turn proves Theorem~\ref{thm:two}.
}
\end{proof}

\begin{proof}(Proof of Lemma \ref{lemma:two2})
Note that, {\col for all time $t_0\le t\le t_1$}, there may be some
injection of packets to a small queue so that its corresponding {\col
set of partial transmissions}
includes some {\col link}s between tall queues {\col
that yields the amount of potential change at least 1.}
We will regard these
kinds of injected packets as ``bad packets'' for the system of small
queues, and we will call these injections  ``bad injections''.
{\col That is, each bad injection in the system of small queues
makes potential change among tall queues at least 1.}
Note that by considering these packets as bad packets, the dynamics of
small queues can be thought as an independent general adversarial
queue system having $n-k$  queues, {\col which means that it is a kind
of subsystem of the original system.} Then essentially, we will show
that the total {\col amount} of these bad injections over all time $t_0\le
t\le t_1$ is bounded by some number which is independent of $t$.
{\col Note that each bad injection in the system of small
queues makes potential change among tall queues at least 1.}

{\col
Now consider all possible cases to obtain the required $t^*$.
At first, we consider the case; (Case I) if there is no bad injection to
small queues for all time $t_0\le t\le t_1$, (Case II)
if there are some bad injections in that time window.

From the definition of the \maxwtbt algorithm,
$r_e(t) \ge  R_{\min}$ for each link $(e,d)$ such
that $q_{v,d}^t-q_{v,d}^t\ge 0$, so we send data along $e$ at least
$R_{\min}$ at once if we can. Without loss of generality, we can assume that
$R_{\min}$ and $R_{\max}$ satisfy $R_{\min}\le\ell_p\le R_{\max}$ for
each $p\in I^W$.
}

\vspace{0.1in}

$-$ {\textbf Case I$~~$} If there is no bad injection to
small queues for all time $t_0\le t\le t_1$,
then {\col by the induction hypothesis}, for all $t_0\le t\le t_1$,
the sizes of small queues are bounded above by
$S_{1}=U(n-k,M_{k+1},0)$. {\col Thus,} the
potential of all the small queues at time $t_1$ is at most ${\coll \frac{\ep}{2}}
L_{1}=(n-k) S_{1}^2$. {\col By Lemma \ref{lemma:zero2}, the decrease
of potential due to a injection to a tall queue is at least
${\coll \frac{\ep}{2}} R_{\min} M_k$, and the decrease of potential due to a transmission
from a tall queue to a small queue at time $t_1$ is at least
$\{(M_{k})^2-(S_{1})^2\}-\{(M_{k}-R_{\min})^2-(S_{1}+R_{\min})^2\}=2R_{\min}(M_k-S_1)$.
Thus, the decrease of potential due to a injection to a tall queue
or the decrease of potential due to a transmission
from a tall queue to a small queue at time $t_1$ is at least
$min\{{\coll \frac{\ep}{2}} R_{\min} M_k,2R_{\min}(M_k-S_1)\}\ge{\coll \frac{\ep}{2}} R_{\min}(M_k-S_1).$}
Note that from the definition of $M_{k}$,{\col
$${\coll \frac{\ep}{2}} R_{\min} (M_{k}- S_{1})\ge {\coll \frac{\ep}{2}} L_{1}+P_0.$$
Therefore,} the decrease of potential due to
an injection to a tall queue or a transmission from a tall queue
to a small queue at time $t_1$ is more than or equal to
the potential of all the small queues at time $t_1$,
{\col and the difference among them is at least $P_0$.
Note also} that the maximum possible increase of the potential induced
by injections during the time $[t_1,t_1+\omega-1]$ is bounded by $P_0$,
and that all the packet movement associated with the injection to a
tall queue at time $t_1$ occurs {\col in this time window of size $\om$.}
Since there was no injection to any of the tall queues during $t_0\le
t\le (t_1-1)$, the potential of the tall queues is non-increasing
for $t_0\le t< t_1$. Hence, by letting $t^*=t_1+\omega-1$, we
{\col have} $P(t^*)\le P(t_0)$.

\vspace{0.1in}

$-$ {\textbf Case II$~~$} {\col Suppose that} there are some bad
injections to small queues. Let $0\le r_1\le r_2\le \ldots \le r_{k-1}$ be the
{\col ordered} list of $(q_1-q_2)$, $(q_2-q_3),\ldots,(q_{k-1}-q_k)$.
{\col As $\{M_1, M_2, \cdots \}$ is a set of queue thresholds,
$\{L_1, L_2, \cdots \}$ defines a set of thresholds for the above list
of queue differences and $\{S_1, S_2, \cdots \}$ gives a bound
on the sizes of the small queues during some period of time
in the following cases. Note that these numbers are independent of $t$.
We can divide (Case II) by following three cases;
(Case II-A): if $r_1>L_{1}$, (Case II-B): if there is  $1\le
m< k-1$ such that for all $1\le j\le m$, $r_j\le L_{j}$, and
$r_{m+1}>L_{m+1}$,
(Case II-C): if $r_m \le L_m$ for all $1\le m\le k-1$.
}

\vspace{0.1in}

$-$ {\textbf Case II-A$~~$} Suppose that $r_1>L_{1}$.
Then any transmission between two tall queues at some
time $t_0< t\le t_1$ will make the decrease of potential more than $
L_{1}$. Let $t^*$ be the smallest time $t^*> t_0$ so that there is a
transmission between two tall queues at time $t^*$. By the induction
hypothesis, for all time $t_0 \le t\le t^*$, the sizes of the small
queues are bounded by $S_{1}=U(n-k,M_{k+1},0)$, and
the potential of the small queues is bounded by ${\coll \frac{\ep}{2}} L_{1}=(n-k)
S_{1}^2$. Then from the same argument as the (Case I), $P(t^*)\le
P(t_0)$.

\vspace{0.1in}

$-$ {\textbf Case II-B$~~$} Suppose that there is  $1\le
m< k-1$ such that for all $1\le j\le m$, $r_j\le L_{j}$, and
$r_{m+1}>L_{m+1}$. We will show that the
potential of all the small queues is bounded by ${\coll \frac{\ep}{2}} L_{m+1}$. We
may assume that bad injections to small queues induce transmissions
just between neighboring tall queues.
Note also that the {\col amount} of bad injections to small queues
during some period of time is bounded by the total {\col amount} of
transmissions between tall queues during that period of time.

We say a {\col link} $e_j=(Q_j,Q_{j+1})$ between two neighboring tall
queues is a {\em tall link} if $q_j-q_{j+1}>L_{m+1}$ and a {\em
small link} otherwise.
{\col
We can divide (Case II-B) by following two cases;
(Case II-B-1) if there is no transmission via tall links for all time $t_0\le t\le t_1$,
(Case II-B-2) if there is a transmission via some tall link for some time $t_0< t\le t_1$.
We will use the following Lemma.
}

\begin{lemma}\label{lemma:three}
Let $r_1,r_2,\ldots,r_m$ be the sizes of the small {\col link}s at time
$t_0$ and assume that $r_j\le L_j$ for all $1\le j\le m$. If there
is no transmission via {\em tall {\col link}s} for $t_0\le t< t'$ and all
the transmissions occur via small {\col link}s, then the total {\col amount} of
packet transmissions via {\em small {\col link}s} during that period of
time is bounded by
{\col $$\frac{m}{2}(r_1+r_2+\ldots+r_m)^2\le \frac{m}{2}(L_1+L_2\ldots+L_m)^2.$$}
\end{lemma}
\begin{proof}
Let $e_{j_1}, e_{j_2},\ldots e_{j_m}$ be the set of small {\col link}s,
where $j_1<j_2<\ldots<j_m$. For $1\le i \le m$, let $s_i$ be
$(q_{j_i}-q_{j_i+1})$. Hence $\{s_i\}_{1\le i\le m}$ is a
permutation of $\{r_i\}_{1\le i\le m}$.

Recall that the sizes of the queues at time $t_0$ are non-increasing
with respect to their indices. Moreover, note that if $j_{i+1}-
j_i\ge 2$ for some $i$, then any packet $p$ that was originally
located at $Q_{m}$, with $m\le j_i+1$ cannot move to $Q_{j_i+2}$ for
all time $t_0\le t\le t'$. Hence we can consider each subset of
consecutive small {\col link}s separately. For example if $j_1,\ldots
,j_m$ are 2,3,5,6,7, then we will consider 2,3 and 5,6,7
separately. Suppose that $j_1,j_2\ldots,j_s$ are consecutive
integers. Since $q_{j_1}^{t'}+\ldots+q_{j_{s+1}}^{t'}
=q_{j_1}^{t_0}+\ldots+q_{j_{s+1}}^{t_0}$, we obtain that $$(q_{j_1}^{t'})^2+\ldots+(q_{j_{s+1}}^{t'})^2 \ge \sum_{i=1}^{s+1}\left(\frac{q_{j_1}^{t_0}+\ldots+q_{j_{s+1}}^{t_0}}{s+1}\right)^2.$$
Thus, the amount of packet transmission via $e_{j_1},\ldots,e_{j_s}$ is
\begin{eqnarray}
&&\{(q_{j_1}^{t_0})^2+\ldots+(q_{j_{s+1}}^{t_0})^2\}
-\{(q_{j_1}^{t'})^2+\ldots+(q_{j_{s+1}}^{t'})^2\}
\nonumber\\
&\le&\{(q_{j_1}^{t_0})^2+\ldots+(q_{j_{s+1}}^{t_0})^2\}
-(s+1)\left(\frac{q_{j_1}^{t_0}+\ldots+q_{j_{s+1}}^{t_0}}{s+1}\right)^2
\nonumber\\
&=&\frac{1}{s+1}\left(s\sum_{i=1}^{s+1}(q_{j_i}^{t_0})^2
-2\sum_{1\le i<k\le s+1}q_{j_i}^{t_0} q_{j_k}^{t_0}\right)
\nonumber\\
&=&\frac{1}{s+1}\left(\sum_{1\le i<k\le s+1}(q_{j_i}^{t_0}-q_{j_k}^{t_0})^2\right)
\nonumber
\end{eqnarray}
\begin{eqnarray}
&\le & \frac{1}{s+1} \binom{s+1}{2} (r_1+\ldots+r_{s+1})^2 \nonumber\\
&= & \frac{s}{2} (r_1+\ldots+r_{s+1})^2 \nonumber\\
&\le & \frac{s}{2} (L_1+\ldots+L_{s+1})^2. \nonumber
\end{eqnarray}
A similar argument holds for other consecutive {\col indices},
separately. Hence the {\col sum of total amount} of transmissions via
small {\col link}s during time $t_0\le t\le t'$ is bounded by
{\col $\frac{m}{2}(L_1+L_2+\ldots+L_m)^2$.}
\end{proof}

$-$ {\textbf Case II-B-1$~~$}  If there is no transmissions via
tall {\col link}s  for all time $t_0\le t\le t_1$. Then by Lemma
\ref{lemma:three}, the total {\col amount} of bad injections  to the
small queues during $t_0\le t\le t_1$ is bounded by
{\col $\frac{m}{2}(L_1+L_2+\ldots+L_m)^2$.
Since each bad injection in the system
of small queues makes potential change among tall queues at least 1,
we conclude that the number of bad packets to the system of
small queues is also at most by $\frac{m}{2}(L_1+L_2+\ldots+L_m)^2$.
Therefore,} for all time $t_0\le
t\le t_1$, the sizes of the small queues are bounded by
{\col $$S_{m+1}=U\left(n-k,M_{k+1},\frac{m}{2}(L_1+L_2+\ldots+L_m)^2\right)$$}
by the induction hypothesis.
Hence, the potential of all the small queues at
time $t_1$ is at most ${\coll \frac{\ep}{2}} L_{m+1}=(n-k) S_{m+1}^2.$

Note that the potential for the tall queues is non-increasing for
$t_0\le t\le t_1$.
{\col By Lemma \ref{lemma:zero2}, the decrease
in potential due to an injection to a tall queue is at least
${\coll \frac{\ep}{2}} R_{\min} M_k$, and the decrease in potential due to a transmission
from a tall queue to a small queue at time $t_1$ is at least
$2R_{\min}(M_k-S_{m+1})$.
Thus, the decrease of potential due to an injection to a tall queue
or the decrease of potential due to a transmission
from a tall queue to a small queue at time $t_1$ is at least
$min\{{\coll \frac{\ep}{2}} R_{\min} M_k,2R_{\min}(M_k-S_{m+1})\}\ge{\coll \frac{\ep}{2}} R_{\min}(M_k-S_{m+1}).$}
Note that from the definition of $M_{k}$,{\col
$${\coll \frac{\ep}{2}} R_{\min}(M_k-S_{m+1}) \ge {\coll \frac{\ep}{2}}{L_{m+1}}+{P_0}.$$}
{\col Therefore, the decrease of the potential} at time $t_1$ is more
than or equal to the potential of all the small queues at $t_1$,
{\col and the difference among them is at least $P_0$.}
By letting $t^*=t_1+\omega-1$, we {\col have} $P(t^*)\le P(t_0)$.\\

$-$ {\textbf Case II-B-2$~~$}  If there is a transmission via some tall
{\col link} for some time $t_0< t\le t_1$, let $t^*$ be the smallest such
$t$. Then similarly, by Lemma \ref{lemma:three}, the total {\col amount} of
bad injections  to the small queues during $t_0\le t\le t^*$ is
bounded by {\col $\frac{m}{2}(L_1+L_2+\ldots+L_m)^2$}.
Hence the sizes of the
small queues during this time interval are bounded by $S_{m+1}$
by the induction hypothesis {\col and from the definition of $t_1$},
so the potential of all the small
queues at time $t^*$ is at most ${\coll \frac{\ep}{2}} L_{m+1}=(n-k) S_{m+1}^2.$
Moreover, during $t_0\le t\le t^*$, for any
tall {\col link} $e_j=(Q_j,Q_{j+1})$, $q_j$ is non-decreasing and
$q_{j+1}$ is non-increasing, {\col because any transmission via
small links can make $q_j$ bigger (when $e_{j-1}$ is a small link),
or $q_{j+1}$ smaller (when $e_{j+1}$ is a small link),
but it cannot increase $q_j-q_{j+1}$. Thus,}
$q_j-q_{j+1}\ge L_{m+1}$ at $ t=t^*$. Hence,
a transmission via a tall {\col link} at time $t^*$ will make the
potential decrease by at least ${\coll \frac{\ep}{2}} L_{m+1}$,
which is more than the potential of all the small queues
at time $t^*$. {\col Note also that} the potential for the tall
queues is non-increasing for $t_0\le t< t^*$.
Hence, we {\col have $P(t^*)\le P(t_0)$.}

\vspace{0.1in}

$-$ {\textbf Case II-C$~~$} {\col Finally, consider the case} when
$r_m \le L_m$ for all $1\le m\le k-1$. Then by Lemma \ref{lemma:three}
and the induction hypothesis, for all time $t_0\le t\le t_1$, the
sizes of small queues are bounded by
{\col $$S_k=U\left(n-k,M_{k+1},\frac{(k-1)}{2}(L_1+L_2+\ldots+L_{k-1})^2\right).$$}
Hence, the potential of all the small queues at time $t_1$ is at most
${\coll \frac{\ep}{2}} L_k=(n-k)S_k^2$.
By Lemma \ref{lemma:zero2}, the decrease
of the potential due to an injection to a tall queue is at least
${\coll \frac{\ep}{2}} R_{\min} M_k$, and the decrease of the potential due to a
transmission from a tall queue to a small queue at time $t_1$ is
at least $2R_{\min}(M_k-S_k)$.
Thus, the decrease of the potential due to a injection to a tall queue
or the decrease of the potential due to a transmission
from a tall queue to a small queue at time $t_1$ is at least
$min\{{\coll \frac{\ep}{2}} R_{\min} M_k,2R_{\min}(M_k-S_k)\}\ge{\coll \frac{\ep}{2}} R_{\min}(M_k-S_k).$
Then, from the definition of $M_{k}$, {\col
${\coll \frac{\ep}{2}} R_{\min}(M_k-S_k) = {\coll \frac{\ep}{2}}{L_k}+{P_0},$} which is more than the
potential of all the small queues at time $t_1$, {\col and the
difference among them is at least $P_0$. Note also that} the
potential of the tall queues is non-increasing for $t_0\le t<t_1$.
Hence, by letting $t^*=t_1+\omega-1$, we {\col have} $P(t^*)\le
P(t_0)$.

\vspace{0.1in}
{\col
Hence in all the cases, we have $P(t^*)\le P(t_0)$ and}
for $t_0\le t\le t^*$, the sizes of small queues are bounded by
$S_j+\omega {\col n R_{\max}}$ for some $1\le j\le k$, so they are bounded by $M_{k}$.
\end{proof}

\section{Characterization of the queue sizes}
We now consider the behavior of the queue sizes under {\coll the} adversarial model. In the case of a stationary stochastic network, the typical ``negative drift'' argument that we described earlier essentially shows that the potential in the system cannot grow much larger than $(k^2 \ve^{-1} (R_{\max})^2)^2$. More precisely, if the potential ever does get larger than that amount then some queue size must be larger than $k^2\ve^{-1}(R_{\max})^2$. At that point the expression for the change in network potential implies the expected drift in potential is non-positive. One consequence of this is that whenever an individual queue size becomes larger than $(k^2 \ve^{-1} (R_{\max})^2)^2$ the expected drift in potential is non-positive.

In contrast, for the \maxwt protocol in the adversarial model the bound on queue size implied by the analysis of Section 4 is actually exponential in the number of users. We now briefly show that this is necessary. In particular, we present an example where the \maxwt protocol does indeed give rise to exponentially-sized queues. Our example is close to an example given in \cite{AndrewsZ-infocom04} in which it was shown that we can get exponential queue sizes in a critically loaded scenario (i.e.\ where $\ve=0$). We now show that this is actually possible in a subcritically loaded example (with $\ve>0$).

We consider a set of $N$ single hop edges (numbered $0,\ldots,N-1$) that are all mutually interfering, i.e. only one edge can transmit data at a time.
Let $a_i(t)$ be the amount of data injected for edge $i$ at time $t$ and let $r_i(t)$ be the edge rate. The adversary defines these quantities in the following simple manner.   At any given time $t$ let $i'=\min\{i:q_i(t) < (1-\ve)2^{i}\}$. If $i'=0$ then the adversary sets $r_0(t)=1$ and $a_0=1-\ve$.  If $i'> 0$ then it sets $r_{i'-1}(t)=1-\ve$, $r_{i'}(t)=\frac{1-\ve}{2}$ and $a_{i'}(t)=\frac{(1-\ve)^2}{2}$.  In both cases all other $r_i(t)$ and $a_i(t)$ values are set to $0$.  It is clear that these definitions are consistent with an $A(1,\ve)$ adversary.

\begin{lemma}
With the above patterns of data arrivals and edge rates, for each $t$ and for each $i$, there exists a $t'\ge t$ such that $q_i(t')\ge (1-\ve)2^{i}$.
\end{lemma}
\begin{proof}
We prove the above statement by induction on $i$.   Suppose that $q_0(t)<1-\ve$. Then for this time step $i'$ is set to $0$ and so $a_0(t)=1-\ve$. Once data has been served for edge $0$ and the arriving data has been added to the edge's queue we have $q_0(t+1)\ge 1-\ve$.  (Note that this assumes that data arrives in a queue after data has been served.  This is a reasonable assumption but if it does not hold then we can simply set $\om\ge 2$ and have all the arrivals in a window of length $\om$ arrive at the beginning of the window.) This completes the base case.

For the inductive step, suppose that $q_i(t)< (1-\ve)2^i$ for an $i>0$. The inductive hypothesis implies that there exists some time $t'\ge t$ at which $i'=i$. Suppose that $t'$ is the first such time step. Between $t$ and $t'$ note that we must have $i'<i$ and so the value of $q_i$ does not change.  When we reach time step $t'$ it must be the case that $q_i(t')<2q_{i-1}(t')$.  Moreover, by the definition of the edge rates $r_{i-1}(t')=1-\ve$ and $r_{i}(t')=\frac{1-\ve}{2}$. Hence the \maxwt protocol serves queue $i-1$ but the arrivals are for queue $i$.  Hence $q_i(t')$ is strictly greater than $q_i(t)$. By repeating this process we eventually reach a time $t''$ at which $q_i(t'')\ge (1-\ve)2^i$.

By the inductive hypothesis there must be a time $t'''\ge t''$ for which $q_j(t''')\ge (1-\ve)2^j$ for all $j\le i-1$.  Between times $t''$ and $t'''$ the value of $q_i$ cannot decrease. Hence at time $t'''$ we have $q_j(t''')\ge (1-\ve)2^j$ for all $j\le i$. The inductive step is complete.
\end{proof}
\begin{corollary}
There exists a network configuration with $N$ edges and an $A(1,\ve)$ adversary such that some queue grows to size $(1-\ve)2^{N-1}$.
\end{corollary}

We remark in conclusion that with a different protocol adversarial models do not necessarily lead to exponentially large queues.  In \cite{AndrewsZ-infocom04} another protocol was presented (which directly keeps track of the past history of edge rates and arrivals) which ensures a maximum queue size of $O(\om k|\mathcal{R}|^2R_{\max})$, where $\mathcal{R}$ is the set of feasible rate values.   However, we still feel that it is of interest to study the performance and stability of the \maxwt protocol in adversarial networks since it is extremely simple to implement and it has been proposed so many times in the literature as a solution to the scheduling problem in wireless networks.

\section{Stability of Approximate Max - weight}
As remarked in the introduction, computing the exact Max-Weight set of feasible transmissions is in general an NP-hard problem.
Hence a natural question to ask is what can be achieved if at each time step we only find an approximate Max-Weight set of feasible transmissions.
In this section we address this question.

Recall that \A$~$assures that there is a {\coll set of fractional movement of packets}
$\Psi_p$ for each $p\in I^{W}$ and there is a edge rate vector
$r_e \in R(t)$ for each $t\in W$, so that each edge is used at most
$(1-\ep)$ times of the sum of rates associated at $e$ during
the time window $W$. Thus, it guarantees that each edge $e$ can transmit more data than is actually required by a ${\coll \frac{1}{1-\ep}}$ factor. Hence, the actual packet movement by \maxwt induces potential changes that are ${\coll \frac{1}{1-\ep}}$ times greater than necessary.

{\col
For an optimization problem,
an $\ep$-approximation algorithm is an algorithm that provides an
approximate solution within $(1\pm\ep)$ factor of the optimal solution.
Although computing the optimal solution of \maxwt is computationally very
hard, in many practical wireless networks $\ep$-approximate solution
for $r(t)$ can be computed in polynomial time. {\blue For example, \cite{Baker94} presented an $\ep$-approximate solution to find the MWIS
(maximum weight independent set) on planar graphs, and this was extended by several authors to more general classes of graphs.} In \cite{GJSS08}, \cite{GJSS09}, an $\ep$-approximate MWIS for a
large class of wireless networks in the Euclidean space is provided.
In our model, we assume an $\ep$-approximate \maxwt computes an
$\ep$-approximate solution $r'(t)$ for each time $t$ so that the potential decrease is
at least $(1-\ep)$ times the maximum {\coll possible} potential decrease at time $t$.
We will prove the stability of any $\hat{\ep}$-approximate \maxwt protocol {\coll under} \A$~$for $\ep>0$, if $0<\hat{\ep}<\ep$.

\begin{theorem}
For $0<\hat{\ep}<\ep$, any $\hat{\ep}$-approximate \maxwtbt is stable under \A.
\end{theorem}}

\begin{proof}
Let $\tilde{\ep}=\frac{\ep-\hat{\ep}}{1-\hat{\ep}}$,
then \A$~$ is $A(\omega,\tilde{\ep})$.
As in the statement of Theorem~\ref{lemma:zero}, if we can associate the injection of
$p\in {\coll I^j \cup I^{j-1}}$ with a set of partial transmissions $\Gamma'_p$, so that
the sum of potential changes due to this injection to a queue of
height $q\ge q^*$ is less than ${\coll -\frac{\tilde{\ep}}{2}} (\beta+1)\ell_p q^\beta$, then all
the other arguments in the proof of Theorem~\ref{thm:two} holds when we replace
$\ep$ with $\tilde{\ep}$.

As in the proof of Theorem~\ref{lemma:zero}, we define
$d'_{p,e}(t')=\frac{1-\ep}{1-\tilde{\ep}}d_{p,e}(t')$ for all $e\in E$,
$t'\in W$, and $p\in {\coll I^j \cup I^{j-1}}$. {\blue
Then we show that $d'_{p,e}(t’)$ satisfy (\ref{eq:c}) if we substitute $\ep$ by $\tilde{\ep}$.
As in the proof of Theorem~\ref{lemma:zero} in \cite{tech}, we define}
\begin{eqnarray}\label{xx}
K'_{e_j}(t')&=&\sum_{i=1}^{m} d'_{p_i,(v_j,u_j)}(t')\left((q^{t'}_{v_j,d_i})^\beta-(q^{t'}_{u_j,d_i})^\beta\right)
\nonumber\\
&=&(1-\hat{\ep})K_{e_j}(t').
\end{eqnarray}

By the definition of $\hat{\ep}$-approximate {\scshape Max-Weight},
we can take $\hat{s}_{e_j}(t')$ for each $e_j=(v_j,u_j)\in E$,
$t'\in W$ such that
\begin{eqnarray}\label{xy}
J'(t')&:=&\sum_{j=1}^{k} \hat{s}_{e_j}(t) \{(q^t_{v_j,d_a^{(e_j)}})^{\beta}-(q^t_{u_j,d_a^{(e_j)}})^{\beta}\}.
\nonumber\\
&\ge& \sum_{j=1}^{k} (1-\hat{\ep}) s_{e_j}(t) \{(q^t_{v_j,d^{(e_j)}})^{\beta}-(q^t_{u_j,d^{(e_j)}})^{\beta}\}
\end{eqnarray}
for some destinations $d_a^{(e_j)}$ for each $e_j.$
By (\ref{xx}) and (\ref{xy}),
we can recursively assign
\begin{equation}
\hat{s}_{p_i,(e_j,d_i)}(t')=min\left\{\frac{{J'}^{((j-1)m+(i-1))}(t')}{(q^{t'}_{v_j,d_i})^{\beta}-(q^{t'}_{u_j,d_i})^{\beta}},\hat{s}_{e_j}^{(i-1)}(t'), \frac{{K'}_{e_j}^{(i-1)}(t')}{(q^{t'}_{v_j,d_i})^\beta-(q^{t'}_{u_j,d_i})^\beta}\right\},
\nonumber
\end{equation}
where $e_j=(v_j,u_j)$ and $d_i$ is the destination of $p_i$,
in the same manner as in the proof of Theorem~\ref{lemma:zero}.
Then, for all
$e_j\in E$, $t'\in W$, and $p_i\in {\coll I^j \cup I^{j-1}}$,
\begin{equation}
\sum_{e\in E}\sum_{p,d}\hat{s}_{p,(e,d)}(t')\left((q^{t'}_{v,d})^{\beta}-(q^{t'}_{u,d})^{\beta}\right)
\le \sum_{e\in E} \hat{s}_{e}(t')\left((q^{t'}_{v,d^{(e)}})^{\beta}-(q^{t'}_{u,d^{(e)}})^{\beta}\right).
\nonumber
\end{equation}
Let $\Gamma'_{p_i}=(\hat{s}_{p_i,e_j}(t'))_{e_j\in E, t'\in W}$ for each $p_i \in {\coll I^j \cup I^{j-1}}$,
we obtain that the sum of potential changes due to this
injection is less than -$\frac{\tilde{\ep}}{2}\ell_{p_i} (\beta+1) q^\beta$ by using the
same argument in section 4.1.
This in turn implies that $\hat{\ep}$-\maxwtbt algorithm is stable under \A.
\end{proof}

\section{Experiments}
\subsection{Simulation Setup}
We now describe a numerical experiment that aims to understand the queue size dynamics of the \maxwt protocol under the adversarial model.
Consider a $n_1 \times n_2$ simple grid
graph $G$, and let $n=n_1n_2$. Then, there are $4n-2n_1-2n_2$ directed edges in the graph.
We assume that all single nodes can be a destination.
We let $n_1=3, n_2=4,$ so $n=12$, and $4n-2n_1-2n_2=34$.

\begin{figure}
\centering
\epsfig{file=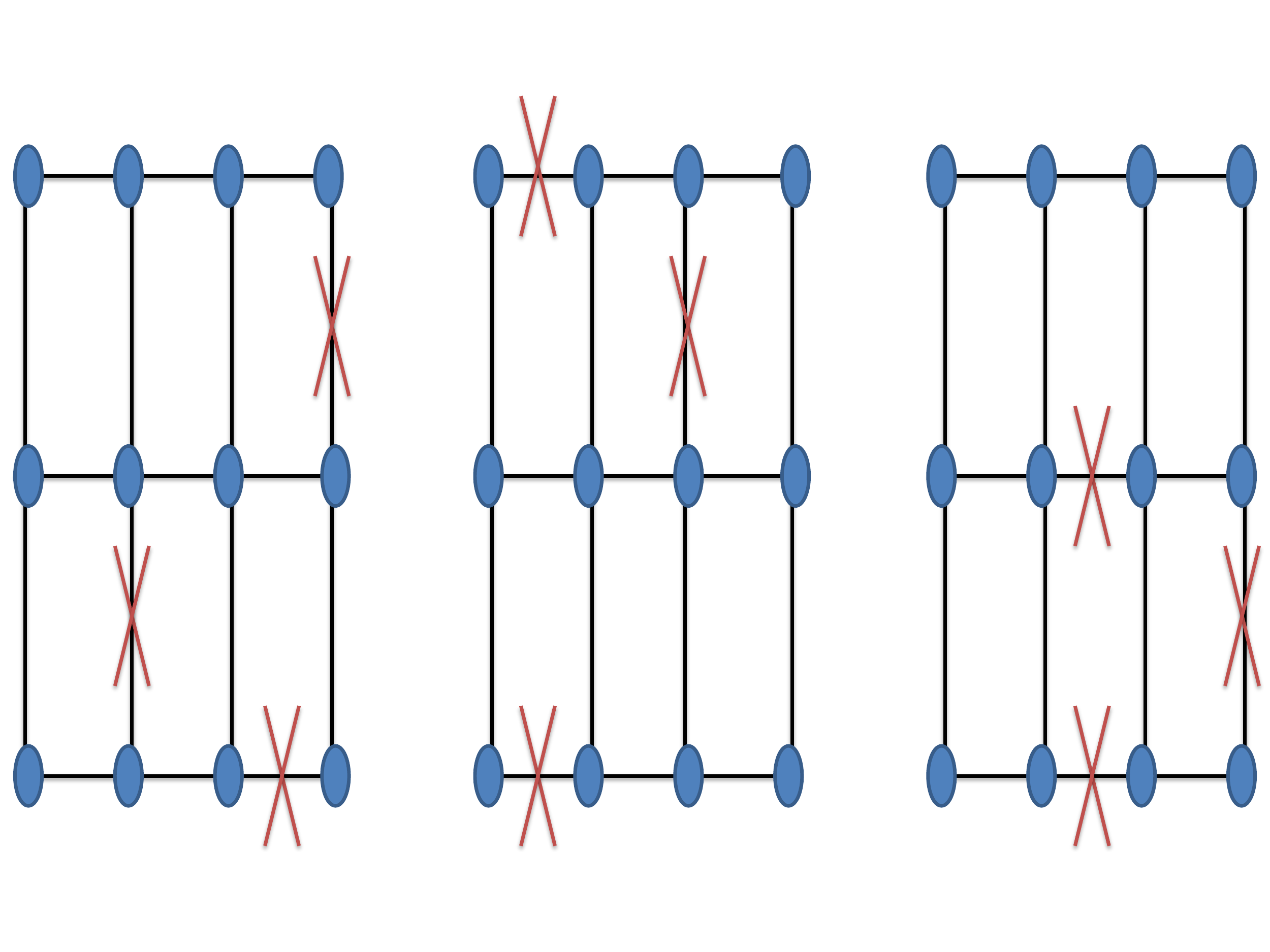, height=0.8in, width=3in}
\caption{The above 3 underlying graphs express which edges are not available under $r^{(1)}, r^{(2)}, r^{(3)}$.
}
\end{figure}

In our simulation, we used 3 different edge rate vectors $r^{(1)},
r^{(2)}, r^{(3)}\in \mathbb{R}^{34}$ for $G$.
For each $r^{(i)}$, $1\le i\le 3$, we select 3 edges among 17 possible edges,
and remove them. The underlying graphs of $r^{(1)}, r^{(2)}, r^{(3)}$
are described in Fig 3. Other directed edges have edge rates chosen
independently and uniformly at random from [0.5,2].
We used the node-exclusive constraint model, i.e., matching constraint model.

Among $n(n-1)$ many distinct source-destination pairs (S-D pairs),
we randomly chose K many S-D pairs $(s_1,d_1),\ldots,(s_K,d_K)$
for $K=10$. When the set of feasible edge rate vectors $R$ is fixed
for all time $t\ge0$, we define the feasible {\em arrival rate} as follows.
The collection of all the feasible arrival rate vectors are called
the network stability region.

\begin{definition}
The {\em arrival rate vector} $\gamma=(\gamma_1,\ldots,\gamma_K)\in[0,1]^K$ corresponding
to the S-D pairs $(s_1,d_1),\ldots,\\(s_K,d_K)$ is said
to be {\em feasible}, if there exist flows, $(f^1,\ldots,f^K)$ such that
\begin{enumerate}
\item For each $1\le j \le K$, $f^j$ routes a flow of at
least $r_j$ from $s_j$ to $d_j$.
\item The induced net flow on the directed edges, $\hat{f}=\sum_{i=1}^K f^j$
belongs to the interior of $co(R)$ where $co(R)$ is the convex hull of $R$.
\end{enumerate}
\end{definition}

If an arrival rate vector is in the interior of $co(R)$,
and the arrivals are identical for all time, then \maxwt is stable \cite{GJSS08}.
Moreover, if an arrival rate vector is in {\coll the interior of} ${co(R)}^c$, then \maxwt is unstable.
We chose $K$ many source-destination pairs at random.
For each $r^{(i)}$, $1\le i\le 3$, we compute 3
different feasible arrival rate vectors that are closed to the boundary
of the network stability region. To do so, we fixed random arrival rate
vectors $\gamma^{(1)}, \gamma^{(2)}, \gamma^{(3)}$ such that each
entry has a value from $[0.5,2]$. We computed constants $c_{ij},$
by binary search, for edge rate vector $r^{(i)},$ and arrival rate vector $\gamma^{(j)}$ so that $c_{ij} \gamma^{(j)}$ is stable under
{\scshape Max-Weight}, and $(c_{ij}+0.001) \gamma^{(j)}$ is not stable under {\scshape Max-Weight},
as described in Fig 4. Each $c_{ij}$ varied from $0.098$ to $0.178$ in our simulation.
We used a sufficiently large time window of size $10^6$ so that we could check the stability.
\begin{figure}
\centering
\epsfig{file=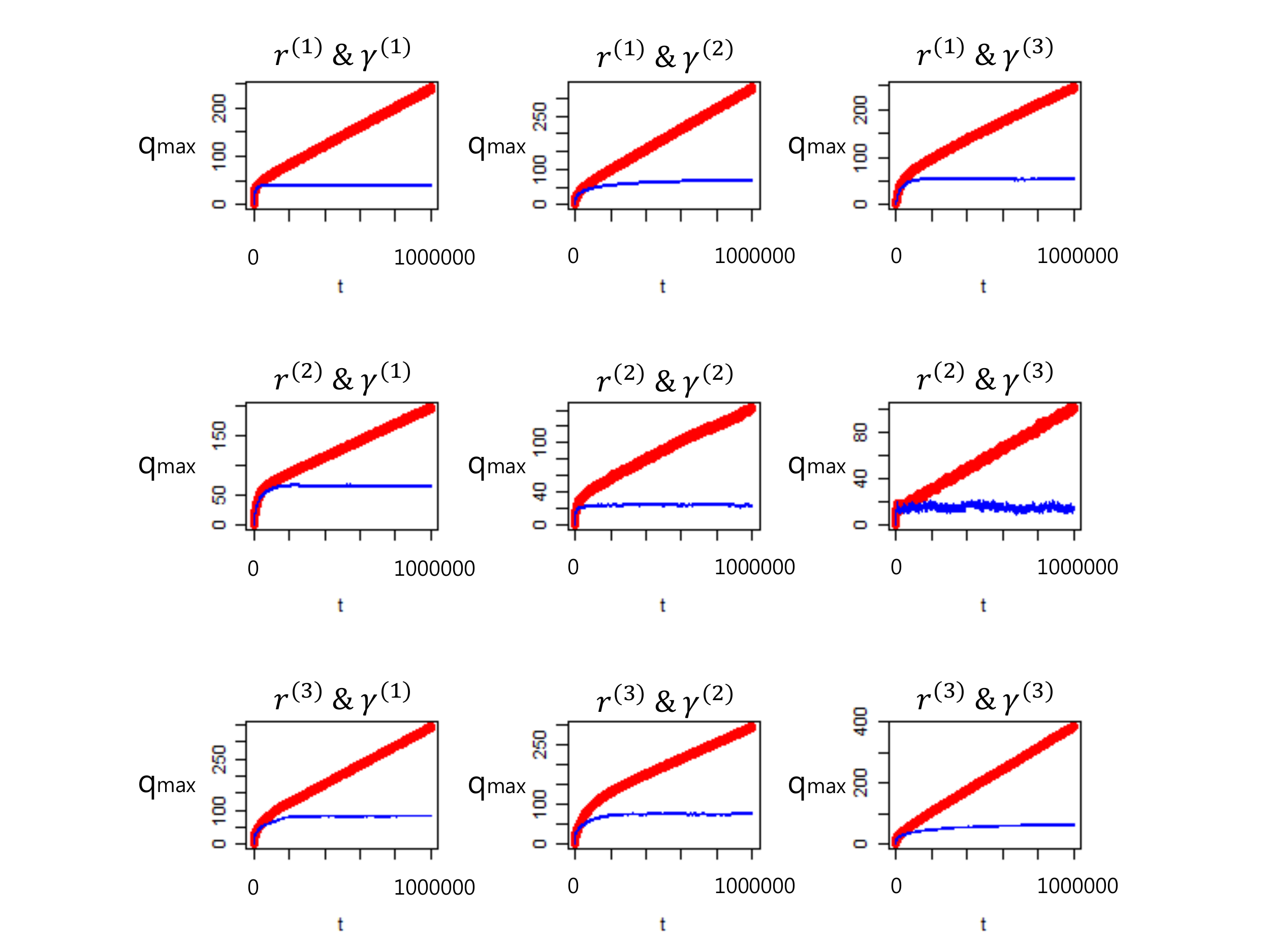, height=3.3in, width=4.4in}
\caption{For each pair of edge rate and arrival rate vector, the plot represents the change of the maximum size of queues for $c_{ij} \gamma^{(j)}$ and $(c_{ij}+0.01) \gamma^{(j)}$ in the time window [0,$10^6$].}
\end{figure}

We did two set of experiments. In both of those experiments, we divided the time $t\ge0$ into non-overlapped sub-windows of ordered phases. The first phase is $t\in[1,\lceil1.5\rceil]$, the second phase is $t\in[\lceil1.5\rceil+1,\lceil1.5+(1.5)^2\rceil],$ and for each $i\ge1$, the $i$th phase is: $t\in[\lceil \sum_{j=1}^{i} (1.5)^{j-1} \rceil+1,\lceil \sum_{j=1}^{i} (1.5)^{j} \rceil].$

In the first experiment, we fixed the edge rate vector $r^{(i)}$ for some $i\in\{1,2,3\}$.
Over time the adversary injects packets as
follows.  For $t\ge0$, if $t$ is in the $j$-th phase,
then inject packets with an arrival rate
{\blue $c_{i\bar{j}}\gamma^{(\bar{j})}$} where $\bar{j}\in\{1,2,3\}$ and $\bar{j}\equiv j$~$(mod$~$3)$.

In the second experiment, over time the adversary determines edge rate vectors and packet arrivals as follows.
For $t\ge0$, if $t$ is in the $i$-th phase,
we assign an edge rate vector $r^{(\bar{i})}$
where $\bar{i}\in\{1,2,3\}$ and $\bar{i}\equiv i$~$(mod$~$3)$,
and we assign an arrival rate vector {\blue $c_{\bar{i}j}\gamma^{(j)}$} at random.

Notice that, in both experiments, the average of the arrival rate vectors
until time T does not converge as T goes to infinity.
Also in the second experiment, the same holds for the edge rate vectors.
However the above injections satisfy the definition of $A(\om,\ep)$ for some $\om>0$ and a small $\ep>0.$
In both setups, we observed
the dynamics of the maximum queue sizes over time.

\begin{figure}
\centering
\epsfig{file=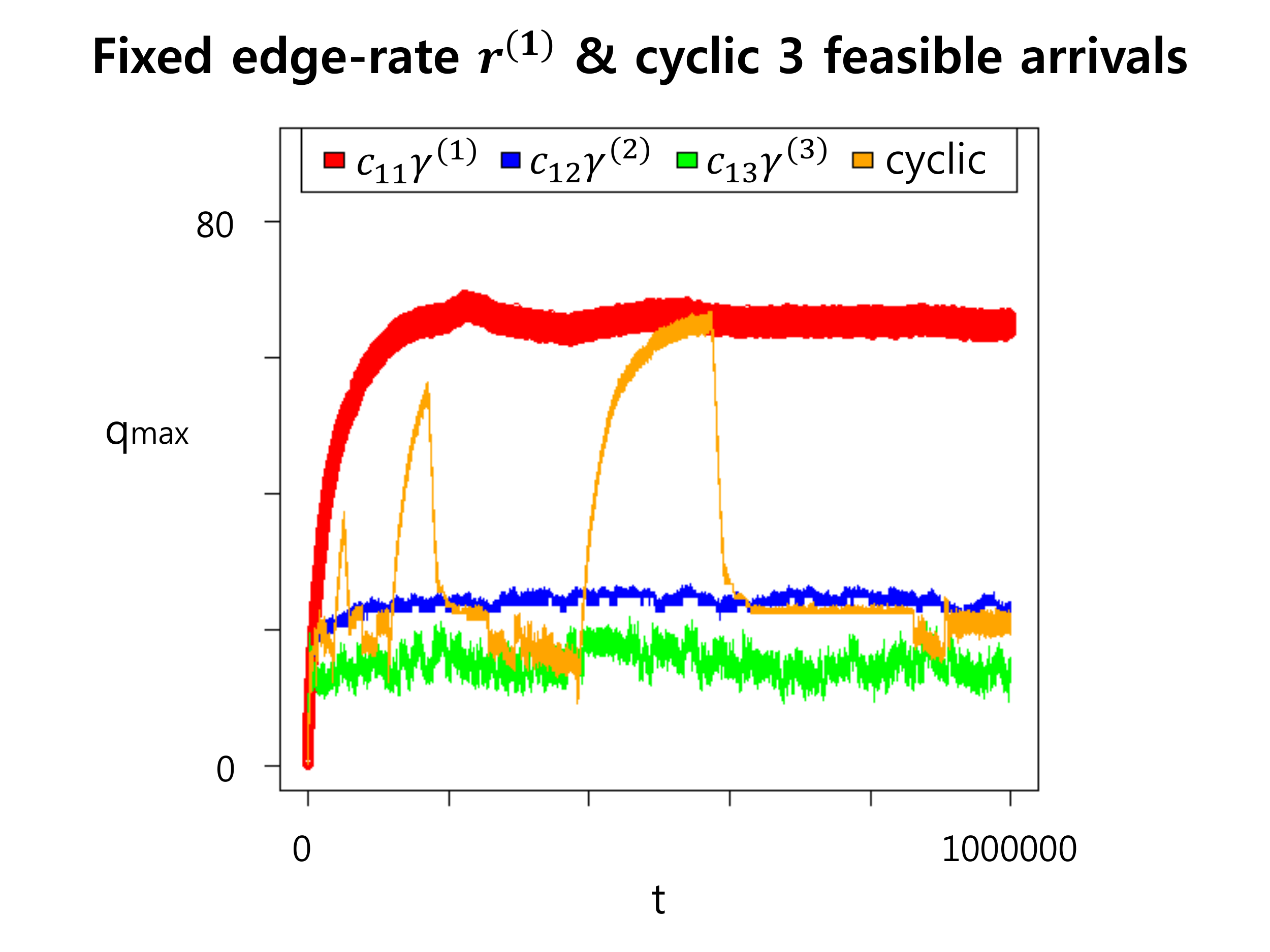, height=3.3in, width=4.4in}
\caption{For the edge rate vector $r^{(1)}$, we plot the maximum queue
size when we use fixed arrival rate vectors {\blue $c_{11}\gamma^{(1)}$,
$c_{12}\gamma^{(2)}$, $c_{13}\gamma^{(3)}$}, and a cyclic arrival rate
vector.}
\end{figure}

\begin{figure}
\centering
\epsfig{file=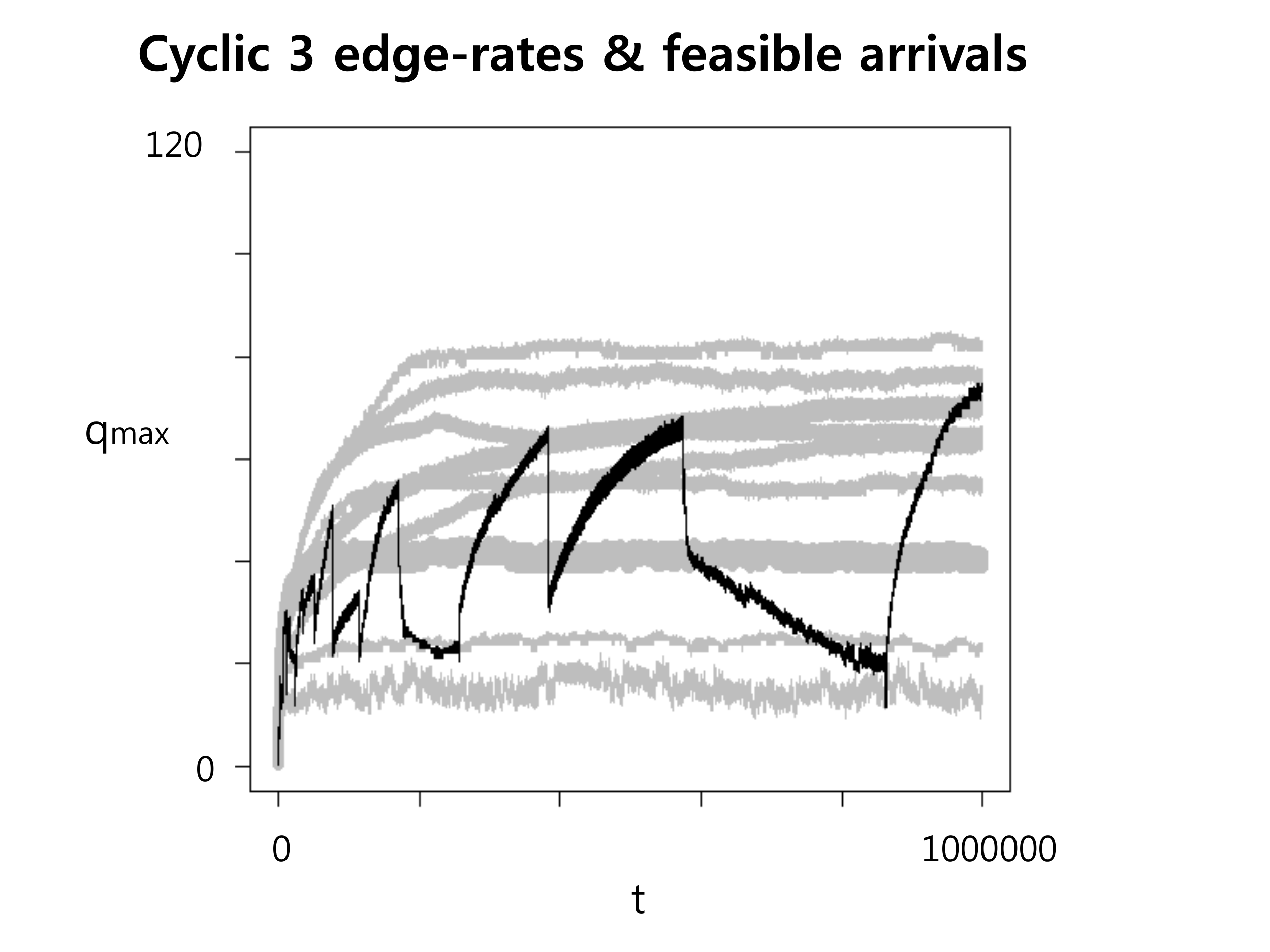, height=3.3in, width=4.4in}
\caption{We use the randomly cyclic edge and arrival rate pairs.
It shows the stability of {\scshape Max-Weight}.}
\end{figure}

\subsection{Simulation Results}
For the first experiment, as Fig 5 shows, for each edge rate vector,
\maxwt is stable with the above cyclic rate vectors. Interestingly, the
maximum queue size may increase in some sub-window,
but it decreases rapidly when the new sub-window starts.
This is because the congested edges are different for each
arrival rate vector, and the traffic-congestions are resolved when the
arrival rate is changed. Notice that the maximum queue sizes for the cyclic rate vector case are bounded above and bounded below by some fixed arrival rate vector cases respectively.

The queue dynamics for the second experiment are described in Fig 6.
{\coll
The gray lines describe queue sizes for fixed edge and arrival rate vectors. The black line describes the queue size for the cyclic rate vector case.}
Again, the maximum queue sizes for the cyclic rate vector case are bounded above and bounded below by some fixed edge and arrival rate vector cases respectively. From our two experiments we observe that
\maxwt make the system stable under $A(\om,\ep)$ even when the edge and arrival rate vectors do not converge over time.

\section{Conclusion}
In this paper we have shown that the \maxwt
protocol remains stable even when the traffic arrivals and {\coll edge rates} are determined in an adversarial manner.

In our opinion the most natural open question concerns the bound on queue size.
Our analysis gives a
bound that is exponential in the network size and we
have shown in Section 5 that such a bound is unavoidable in the
general case.  However, achieving these large queue sizes involves
choosing the achievable rate vectors $R(t)$ in a very specific
manner.  We are interested in whether there are any simple sufficient
conditions on the sets $R(t)$ which would ensure that such large
queues do not occur.

\bibliographystyle{plain}
\bibliography{mybibfile}

\appendix
\section{{\coll Remaining proof of Theorem~\ref{lemma:zero}}}

\begin{proof}
Now, for $t,t'\in W$, $|q^t_{u,d}-q^{t'}_{u,d}|\le n R_{\max}\omega$
since at each time slot, at most $R_{\max}$ amount of data can move along a {\col link from} $u$. Hence by considering $\omega$ and
$n$ and $R_{\max}$ as constants, we obtain that for any $t,t'\in W$,
$(q^{t'}_{u,d})^\beta=(q^t_{u,d})^\beta+O\left((q^{t}_{u,d})^{\beta-1}\right).$

{\col Suppose that a packet $p$ with size $\ell_p$ is injected at a
node $v_0$ at time $t_0\in W$. Let $d$ be the
destination of $p$. Then, the potential change due to the injection of $p$ is,
\begin{eqnarray}
&& \sum_{(x,y)\in E} \sum_{t'\in W} s_{p,((x,y),d)}(t')(\beta+1)
\left|(q^{t_0}_{x,d})^\beta-(q^{t_0}_{y,d})^\beta
+O\left((q^{t_0}_{x,d})^{\beta-1}+(q^{t_0}_{y,d})^{\beta-1}\right)\right|
\nonumber\\
&=& \sum_{e=(v,u)\in\Psi_p} \sum_{t'\in W} d_{p,e}(t')(\beta+1)
\left|(q^{t_0}_{v,d})^\beta-(q^{t_0}_{u,d})^\beta
+O\left((q^{t_0}_{v,d})^{\beta-1}+(q^{t_0}_{u,d})^{\beta-1}\right)\right|
\nonumber\\
&\ge& \sum_{e=(v,u)\in\Psi_p}
\frac{1}{1-\ep} \{ \sum_{t'\in W} \ell(p,e,t') \}(\beta+1)
\left|(q^{t_0}_{v,d})^\beta-(q^{t_0}_{u,d})^\beta
+O\left((q^{t_0}_{v,d})^{\beta-1}+(q^{t_0}_{u,d})^{\beta-1}\right)\right|
\nonumber\\
&\ge& \sum_{e=(v,u)\in\Psi_p}
\frac{1-\ep/2}{1-\ep}\ell_p(\beta+1)
\left|(q^{t_0}_{v,d})^\beta-(q^{t_0}_{u,d})^\beta
+O\left((q^{t_0}_{v,d})^{\beta-1}+(q^{t_0}_{u,d})^{\beta-1}\right)\right|
\nonumber\\
&\ge& \sum_{e=(v,u)\in\Psi_p}
\frac{\ell_p}{1-\ep/2}(\beta+1)
\left|(q^{t_0}_{v,d})^\beta-(q^{t_0}_{u,d})^\beta
+O\left((q^{t_0}_{v,d})^{\beta-1}+(q^{t_0}_{u,d})^{\beta-1}\right)\right|
\nonumber\\
&\ge& \frac{1}{1-\ep/2}\ell_p(\beta +1)(q^{t_0}_{v_0,d})^\beta +\ell_p
O\left((q^{t_0}_{v_0,d})^{\beta-1}\right).
\nonumber
\end{eqnarray}}
because $\{\sum_{t'\in W} \ell(p,e,t')\}\ge (1-\frac{\ep}{2})\ell_p$ holds for all packet $p$ and edge $e$.

The increase of
potential due to the direct injection of $p$ is $\ell_p
(\beta+1)(q^{t_0}_{v_0,d})^\beta + \ell_p
O((q^{t_0}_{v_0,d})^{\beta-1})$. Hence the total change of potential
induced by this injection {\col of a packet $p$ is,
\begin{eqnarray}
&&\ell_p (\beta+1)(q^{t_0}_{v_0,d})^\beta +\ell_p O\left((q^{t_0}_{v_0,d})^{\beta-1}\right)
\nonumber\\
&-&\sum_{(x,y)\in E} \sum_{t'\in W} s_{p,((x,y),d)}(t')(\beta+1)
\left|(q^{t_0}_{x,d})^\beta-(q^{t_0}_{y,d})^\beta
+O\left((q^{t_0}_{x,d})^{\beta-1}+(q^{t_0}_{y,d})^{\beta-1}\right)\right|
\nonumber\\
&\le&
-\frac{\ep/2}{1-\ep/2}\ell_p(\beta +1)(q^{t_0}_{v_0,d})^\beta +\ell_p
O\left((q^{t_0}_{v_0,d})^{\beta-1}\right).
\nonumber
\end{eqnarray}}
Hence there is a constant $q^*$, depending on $n$, $\omega$ and
$\ep$, so that if $q\ge q^*$ the sum of potential changes due to the
injection is less than $-\frac{\ep}{2}\ell_p q^\beta$.
\end{proof}

\end{document}